\documentclass[11pt, reqno]{amsart}
\usepackage{amssymb}
\usepackage{amsmath}
\allowdisplaybreaks[1]
\usepackage{graphicx,color}
\usepackage{wrapfig,framed}
\usepackage{mathtools}
\usepackage[height=23cm, width=16.9cm, hmarginratio={1:1}]{geometry}
\usepackage{hyperref}

\theoremstyle{plain}
\newtheorem{theorem}{Theorem}
\newtheorem{prop}{Proposition}
\newtheorem{lemma}{Lemma}
\newtheorem{cor}{Corollary}
\theoremstyle{definition}
\newtheorem{defi}{Definition}
\newtheorem{remark}{Remark}
\newtheorem{example}{Example}

\newcommand{\beq}{\begin{equation}}
\newcommand{\eeq}{\end{equation}}
\newcommand{\nn}{\nonumber}

\newcommand{\bllb}{\bigl(\hskip -0.05truecm \bigl(}
\newcommand{\brrb}{\bigr)\hskip -0.05truecm \bigr)}
\newcommand{\llb}{(\hskip -0.05truecm (}
\newcommand{\rrb}{)\hskip -0.05truecm )}
\newcommand{\llm}{[\hskip -0.04truecm [}
\newcommand{\rrm}{]\hskip -0.04truecm ]}
\newcommand{\bllm}{\bigl[\hskip -0.04truecm \bigl[}
\newcommand{\brrm}{\bigr]\hskip -0.04truecm \bigr]}

\newcommand{\cS}{\mathcal{S}}
\newcommand{\cL}{\mathcal{L}}

\newcommand{\QQ}{{\mathbb Q}}
\newcommand{\NN}{{\mathbb N}}
\newcommand{\CC}{{\mathbb C}}

\newcommand{\ZZ}{{\mathbb Z}}

\newcommand{\R}{R}
\newcommand{\A}{{\mathcal A}}

\newcommand{\M}{{\mathcal M}}
\newcommand{\pf}{\noindent{\it Proof \ }}
\newcommand{\tr}{{\rm tr}}
\newcommand{\bt}{{\bf t}}
\newcommand{\res}{{\rm res}}
\newcommand{\bdzero}{{\bf 0}}

\newcommand{\p}{\partial}

\newcommand{\g}{{\mathfrak g}}
\newcommand{\Ker}{{\rm Ker}}
\newcommand{\Image}{{\rm Im}}
\newcommand{\ad}{{\rm ad}}
\newcommand{\z}{\zeta}
\newcommand{\epf}{$\quad$\hfill\raisebox{0.11truecm}{\fbox{}}\par\vskip0.4truecm}

\newcommand{\Ztheta}{Z_{\textsc{\tiny $\Theta$}}}
\newcommand{\utheta}{u_{\textsc{\tiny  $\Theta$}}}
\newcommand{\utc}{u_{\textsc{\tiny  $\Theta(C)$}}}
\newcommand{\uteight}{u_{\textsc{\tiny  $\Theta(1/8)$}}}
\newcommand{\tauC}{\tau_{\textsc{\tiny  $\Theta(C)$}}}
\newcommand{\DzwC}{D_{\textsc{\tiny  $\Theta(C)$}}(z,w)}
\newcommand{\DC}{D_{\textsc{\tiny  $\Theta(C)$}}}
\newcommand{\taueight}{\tau_{\textsc{\tiny  $\Theta(1/8)$}}}
\newcommand{\otc}{\Omega^{\textsc{\tiny  $\Theta(C)$}}}
\newcommand{\oteight}{\Omega^{\textsc{\tiny  $\Theta(1/8)$}}}

\def\sm#1#2#3#4{\bigl(\smallmatrix#1&#2\\#3&#4\endsmallmatrix\bigr)}

\usepackage{array}
\newcolumntype{M}[1]{>{\centering\arraybackslash}m{#1}}
\newcolumntype{R}[1]{>{\raggedleft\arraybackslash}m{#1}}
\newcolumntype{N}{@{}m{0pt}@{}}

\def\={\; = \;}
\def\+{\; + \;}
\def\:={\; := \; }

\def \m {\,-\,}
\def \wt {\widetilde}

\newcommand{\indicationfootnote}{\thanks}

\begin{document}
\title[On tau-functions for the KdV hierarchy]{On tau-functions for the KdV hierarchy}
\author[Dubrovin]{Boris Dubrovin$^{\dagger}$}
\indicationfootnote{$^{\dagger}$Deceased on March 19, 2019.}
\author[Yang]{Di Yang}
\author[Zagier]{Don Zagier}
\date{}
\begin{abstract} 
For an {\it arbitrary} solution to the KdV hierarchy, the generating series of logarithmic derivatives of the tau-function 
of the solution can be expressed by the basic matrix resolvent via algebraic manipulations. 
Based on this we develop in this paper two new formulae 
for the generating series by introducing a pair of 
wave functions of the solution. Applications to the Witten--Kontsevich tau-function, to the generalized 
Br\'ezin--Gross--Witten (BGW) tau-function, as well as to 
a modular deformation of the generalized BGW tau-function which we call the Lam\'e tau-function 
are also given.

\medskip


\smallskip 

\end{abstract}

\maketitle

\setcounter{tocdepth}{1}
\tableofcontents

\section{Introduction}
To make this article accessible also to non-specialists, we recall the definitions of all of
the main objects studied (KdV hierarchy, tau-function, matrix resolvents, wave functions, \dots), 
in some cases with definitions slightly different from the most standard ones. Experts can simply skip this material.

\subsection{The Korteweg--de Vries hierarchy}\label{section11} 
The Korteweg--de Vries (KdV) equation  
\beq \label{KdVeq}
u_t \= u \, u_x \+ \frac{1}{12} \, u_{xxx} \,, 
\eeq
discovered in the study of water waves in the 19th century, extends to a {\it hierarchy} of PDEs, 
\beq\label{KdVhk}
\frac{\p u}{\p t_k} \= \frac{u^k}{k!} u_x \+ Q_k(u,u_x,u_{xx},\dots)   \qquad (k\geq 0) \,, 
\eeq
where $u$ is now a function of the multivariable $\bt=(t_0=x, t_1=t, t_2, t_3,\dots)$ and the $Q_k$ are 
specific polynomials with $Q_0=0$ and $Q_1= u_{xxx}/12$. We recall a simple construction of this hierarchy. 
Denote by $\A$ the polynomial algebra $\QQ[u_0,u_1,u_2,\dots]$ with the grading $\deg u_i=i$, and by~$\p$ the 
derivation on~$\A$ that satisfies $\p(u_i)=u_{i+1}$.  For arbitrary elements $Q_0,Q_1,Q_2,\dots$ in~$\A$ we can 
define uniquely a family of derivations $D_k$ ($k\geq0$) on~$\A$ by 
\beq
D_0\= \p \,, \quad [D_k,D_0] \= 0\, , \quad D_k(u_0) \= \frac{u_0^k}{k!} \, u_1 \+ Q_k \,. \label{KdVformal}
\eeq
It turns out that if we require that $Q_k\in \A^{\geq 2}$ with $Q_0=0$ and $Q_1=u_3/12$ and that all $D_k$ commute with~$D_1$, then 
the polynomials $Q_k$ exist and are unique. The first few values are
\begin{align}
& Q_0\=0\,, \quad Q_1\=\frac1{12} u_3\,, \quad Q_2\=\frac{1}{12} (2 u_1 u_2 +  u_0 u_3) \+ \frac{1}{240} u_5\,, \nn \\
& Q_3\=\frac{1}{24} \bigl(u_1^3+4  u_0 u_1 u_2+u_0^2u_3\bigr) \+ \frac{1}{240}  (5u_2u_3+3 u_1  u_4+u_0 u_5) \+ \frac{1}{6720}  u_7 \,, \nn
\end{align}
and all the $D_k$ commute. (See~\cite{Dickey,LZ}. 
A new proof of the existence and pairwise commutativity is given in Section~\ref{section2}.) 
We call these unique derivations $D_k$ the KdV derivations. The formal system~\eqref{KdVformal} called the abstract KdV hierarchy leads to the 
compatible system of PDEs~\eqref{KdVhk} if $u=u(\bt)$ and we identify $u_i$ and $D_k$ with $\p_x^i(u)$ and $\p/\p t_k$, respectively. 

Let $V$ be a ring of functions of~$x$ closed under $\p_x$. (Usually $V$ will be $\CC\llm x\rrm$ or $\CC\llb x \rrb$.) 
For given initial data $f(x) \in V$, one can solve the KdV hierarchy~\eqref{KdVhk} to get a   
unique solution $u=u(\bt)$ in $V \llm \bt_{>0}\rrm$ with $u(x,0,0,\dots) =f(x)$. This gives a 1-1 correspondence: 
\beq\label{11correspondence}
\mbox{$\bigl\{$solution $u$ of~\eqref{KdVhk} in $V \llm \bt_{>0}\rrm  \bigr\}$ \quad  $\longleftrightarrow$ \quad  $\bigl\{$initial data $ f \bigr\} \= V$}\, .
\eeq
Below we give three cases of particular interest. More examples can be found in~\cite{BDY1,KMZ}.

\begin{example}  $f(x)=x$. The corresponding unique solution~$u$ in~$\CC\llm x, \bt_{>0} \rrm$ governs~\cite{Kont,Witten} the integrals 
\beq\label{wkcor}
\int_{\overline{\mathcal{M}}_{g,n}} \psi_1^{p_1} \cdots \psi_n^{p_n} \,, 
\eeq
where $\overline{\mathcal{M}}_{g,n}$ denotes the Deligne--Mumford moduli space of 
stable algebraic curves of genus~$g$ with $n$ distinct marked points and 
$\psi_j$ denotes the so-called $\psi$-class (see Section~\ref{section4} for the precise definition). This solution is often called the {\it Witten--Kontsevich 
solution}, denoted by $u_{\rm WK}$. 
\end{example}

\begin{example}\label{example2} $f(x)=\frac{C}{ (x-1)^2}$, $C\in \CC$. 
The corresponding solution~$u$ in~$\CC\llm x,\bt_{>0}\rrm$  
will be called the {\it generalized BGW solution}~\cite{A,MMS}, denoted by~$u_{\Theta(C)}$. 
For $C=1/8$, $u_{\Theta(C)}$ governs~\cite{N} the integrals 
\beq
\int_{\overline{\mathcal{M}}_{g,n}} \Theta_{g,n} \, \psi_1^{p_1} \cdots \psi_n^{p_n}  \,, 
\eeq
called the $n$-point $\Theta$-class intersection numbers. Here $\Theta_{g,n}$ denotes the Norbury $\Theta$-class~\cite{N} whose 
precise definition will be recalled in Section~\ref{section4}. We call $u_{\Theta(1/8)}$ the {\it BGW-Norbury solution}. 
\end{example}

\begin{example}\label{example3}
$f(x)=C \, \wp(x; \tau)$, $C\in \CC$, $\tau\in \mathfrak{H}$ = upper half plane. 
Here we can take $V$ to be the ring $\CC[g_2,g_3,\wp,\wp']/(\wp'^2- 4 \wp^3+g_2 \wp + g_3)$
with $g_2 = 60 \, G_4$ and $g_3 = 140 \, G_6$, 
where $G_{2k}= \sum_{(m,n)\in\ZZ^2 \smallsetminus (0,0)} \frac1{(m+n\tau)^{2k}}$,~$k\geq 2$.
The corresponding unique solution~$u$ in $V\llm \bt_{>0}\rrm $ 
is a modular deformation of~$u_{\Theta(C)}$, denoted by $u_{\rm elliptic}$ and discussed in more detail  
 in Section~\ref{section5}. We call~$u_{\rm elliptic}$ a {\it Lam\'e solution} of the KdV hierarchy. 
\end{example}

\subsection{The tau-function and the matrix resolvent approach}\label{MRI} 
Let $(D_k)_{k\geq 0}$ be the KdV derivations 
defined above. 
By a {\it tau-structure} for the abstract KdV hierarchy we mean a collection of elements 
$\Omega_{p,q}\in\mathcal{A}$~$(p,q\geq 0)$ satisfying
\beq\label{taustructure}
 \Omega_{0,0} \= u_0\,,  \quad   \Omega_{p,q} \= \Omega_{q,p} \,,   \quad    D_r  \bigl(\Omega_{p,q}\bigr) \=  D_q \bigl(\Omega_{p,r}\bigr) \,,   \qquad 
 \forall\, p,q,r\geq 0\,. 
\eeq
(A more general abstract tau-structure allows $\Omega_{0,0}$ to be some other element of~$\A$ with non-degenerate $0^{\rm th}$-order term, but this will not be studied in this paper.) 
Since the $D_k$ commute, 
the elements
\beq \label{tausymm} \Omega_{p_1,\dots,p_n} \:= D_{p_1} \cdots D_{p_{n-2}} \, (\Omega_{p_{n-1}, p_n})
\qquad (n\geq 3,\; p_1,\dots,p_n\geq 0) 
\eeq
are symmetric in their indices.  
One can show that the tau-structure exists
and is unique up to replacing $\Omega_{p,q}$ by $\Omega_{p,q}+c_{p,q}$, where $c_{p,q}=c_{q,p}$ are constants.  
By the {\it canonical} tau-structure we will mean the unique choice for which all $\Omega_{p,q}$ vanish when all $u_i=0$, the first few values being
$$ \Omega_{0,1}\,=\,\frac{u_0^2}2 +\frac{u_2}{12}\,, \quad \Omega_{0,2}\,=\, \frac{u_0^3}6 + \frac{u_1^2}{24}  + \frac{u_0u_2}{12} +\frac{u_4}{240} \,, 
\quad \Omega_{1,1}\,=\, \frac{u_0^3}3  + \frac{u_1^2}{24} +\frac{u_0u_2}6 +\frac{u_4}{144}\,,\quad\cdots \,.  $$

If $u=u(\bt)$ is a solution of~\eqref{KdVhk} and we write $\Omega_{p,q}(\bt)$ as the image of $\Omega_{p,q}$ under $u_i\mapsto \p_x^i(u)$, 
then~\eqref{taustructure} implies that there exists a function $\tau=\tau(\bt)$, a so-called {\it tau-function}  of the solution~$u$, such that  
\beq\label{taufunctionOmega}
\Omega_{p,q}(\bt) \= \frac{\p^2 \log \tau(\bt)}{\p t_p \p t_q } \qquad (p\,,q\geq 0) \,.
\eeq
The symmetry in~\eqref{tausymm} 
is then obvious, since the image $\Omega_{p_1,\dots,p_n}(\bt)$ of $\Omega_{p_1,\dots,p_n}$ under $u_i\mapsto \p_x^i(u)$ is
\beq
\Omega_{p_1,\dots,p_n}(\bt) \= \frac{\p^n \log \tau(\bt)}{\p t_{p_1} \cdots \p t_{p_n}} \; \qquad (n\geq 2,\; p_1,\dots,p_n\geq 0) \,. \label{Taylorlogtau}
\eeq
Denote $\Omega_p(\bt):=\p_{t_p} \bigl(\log \tau(\bt)\bigr)$. 
These logarithmic derivatives are called {\it $n$-point correlation functions} of~$u$  
 and the specializations 
$\Omega_{p_1,\dots,p_n}(x) :=\Omega_{p_1,\dots,p_n}(x,\bdzero)$ are called {\it $n$-point partial correlation functions} of~$u$, while
the evaluation of~$\Omega_{p_1,\dots,p_n}(x)$ at a particular value of~$x$ will be called an {\it $n$-point correlator} of~$u$.
If we choose the canonical tau-structure as defined above, we will call~$\tau$  {\it the} tau-function of~$u$, although
it is still defined up to multiplying the exponential of an arbitrary linear function of~$\bt$. 
This definition of the tau-function~$\tau$ is known to agree with other literature (\cite{DKJM,DZ-norm}). 

The key fact for this paper is that the canonical tau-structure for the abstract KdV hierarchy can be given by an explicit generating 
series using the so-called matrix resolvent (MR) approach~\cite{BDY1,BDY3}. This goes as follows.  Let $\cL$ (``matrix Lax operator") 
be the operator $\p+ \Lambda(\lambda) + q$, 
where $\Lambda(\lambda) = \sm 0 1 \lambda 0$, $q = \sm00{-2u_0}0$. 
Let $\cS={\rm sl}_2(\CC) \bllb\lambda^{-1}\brrb $
be the space of ${\rm sl}_2(\CC)$-valued formal Laurent series in~$\lambda^{-1}$. 
The {\it principal grading} on $\mathcal{A} \otimes \cS$ is defined by 
$\deg E = 1$, $\deg F=-1$, $\deg \lambda = 2$, $\deg u_i=0$, where $E=\sm0100$, $F=\sm0010$ are the Weyl generators. 
The {\bf basic matrix resolvent} of~$\cL$ is defined as the 
unique element $R(\lambda)=R(\lambda; u_0,u_1,\dots)\in \cS\otimes \mathcal{A}$ satisfying
\beq\label{defiR}
\bigl[ \cL \,, R(\lambda) \bigr] \= 0\,, \qquad R(\lambda) \= \Lambda(\lambda) \+ \mbox{l.o.t.}\,, 
\qquad  {\rm Tr} \, R(\lambda)^2 \= 2 \lambda\,,
\eeq
where ``l.o.t" means lower order terms with respect to the principal gradation. 
\noindent The proof of existence and uniqueness of~$R(\lambda)$ can be found in~\cite{BDY3} or in Section~\ref{section2} of the present paper.
The upper right entry $R(\lambda)_{12}$ of the basic matrix resolvent which 
we denote by~$b(\lambda)$ plays an important role.
The following proposition, originally proved in~\cite{BDY1}, will be re-proved in Section~\ref{section2} in a more straightforward way.

\begin{prop}[\cite{BDY1,BDY3, Zhou1}]\label{taulemma}
(i) The differential polynomials $\Omega_{p,q}\in\A$ defined by the generating series
\beq
\frac{{\rm tr} \, \bigl(R(\lambda) R (\mu)\bigr)}{(\lambda-\mu)^2} \,-\, \frac{\lambda+\mu}{(\lambda-\mu)^2} 
  \=  \sum_{p,q\geq 0} \frac{(2p+1)!! \, (2q+1)!!}{\lambda^{p+1} \mu^{q+1}} \, \Omega_{p,q}  
 \label{tauom}
\eeq
form a tau-structure for the abstract KdV hierarchy, vanishing at $u_0=u_1=\cdots=0$.   

\quad\noindent (ii) For any integer $n\geq 3$ and the corresponding $\Omega_{p_1,\dots,p_n}$ as defined in~\eqref{tausymm} we have 
\begin{align}
& \sum_{p_1,\dots,p_n\geq 0} \Omega_{p_1,\dots,p_n} \, \prod_{j=1}^n \frac{(2p_j+1)!!}{\lambda_j^{p_j+1}}  
  \=  -\sum_{\sigma\in S_n/C_n} \frac{ {\rm tr} \,\bigl(R(\lambda_{\sigma(1)})\cdots R(\lambda_{\sigma(n)})\bigr) } 
  {\prod_{i=1}^n \bigl(\lambda_{\sigma(i+1)}-\lambda_{\sigma(i)}\bigr)} \,,  \label{thmformula}
\end{align}
where $S_n$ denotes the symmetry group, $C_n$ the cyclic group, and it is understood that $\sigma(n+1)=\sigma(1)$.
\end{prop}

Just as with $\Omega_{p,q}$ and $\Omega_{p,q}(\bt)$, we can evaluate the abstract basic matrix resolvent $R(\lambda)$ for a solution $u(\bt)$ of 
the KdV hierarchy by substituting $\p_x^i(u)$ for~$u_i$. The resulting matrix-valued function $R(\lambda,\bt)$ is called the {\it basic matrix resolvent} of~$u$. Denote for short $\R(\lambda, x)=R(\lambda,x,\bdzero)$.

\begin{cor}[\cite{BDY1}] \label{cor2}   
For any $n\geq 2$, the following formulas hold true:
\begin{align}
& \sum_{p_1,\dots,p_n} \Omega_{p_1, \dots,p_n}(\bt) \, \prod_{j=1}^n \frac{(2p_j+1)!!}{\lambda_j^{p_j+1}}  \= 
-\sum_{\sigma\in S_n/C_n} 
\frac{{\rm tr} \, R(\lambda_{\sigma(1)}, \bt)\cdots R(\lambda_{\sigma(n)}, \bt)}
{\prod_{i=1}^n \bigl(\lambda_{\sigma(i+1)}-\lambda_{\sigma(i)}\bigr)} 
\,- \, \frac{\lambda_1+\lambda_2}{(\lambda_1-\lambda_2)^2}\delta_{n2} \,,  \label{npointtau} \\
& \sum_{p_1,\dots,p_n} \Omega_{p_1, \dots,p_n}(x) \, \prod_{j=1}^n \frac{(2p_j+1)!!}{\lambda_j^{p_j+1}}  \= 
-\sum_{\sigma\in S_n/C_n} 
\frac{{\rm tr} \, \R(\lambda_{\sigma(1)},x)\cdots \R(\lambda_{\sigma(n)},x)}
{\prod_{i=1}^n \bigl(\lambda_{\sigma(i+1)}-\lambda_{\sigma(i)}\bigr)} 
\,- \, \frac{\lambda_1+\lambda_2}{(\lambda_1-\lambda_2)^2}\delta_{n2} \,. \label{npointtaux}
\end{align}
\end{cor}
\noindent We note that $\R(\lambda, x)$ is equal to $R(\lambda)$ with $u_i$ replaced by $\p_x^i(f)$, where $f$ is the initial data of the solution. Therefore, 
formula~\eqref{npointtaux} produces all higher order logarithmic derivatives of the tau-function of~$u$ at $\bt_{>0}= \bdzero$ with the knowledge of the initial data~$f$.

\begin{remark} 
Tau-structures are among the most important notions in the theory of 
integrable systems (see e.g.~\cite{BBT, BDGR, DKJM, Dickey, Du1, DZ-norm, Hirota,  JMU, S, SW}). 
For the KdV hierarchy, it is known that there are several different but equivalent definitions for 
the canonical one~\cite{BDY1, Dickey}; however, one can also define 
essentially different tau-structures by using normal Miura transformations~\cite{DLYZ,DZ-norm} for instance.
\end{remark}

\subsection{From wave functions to $n$-point correlation functions}\label{Introwave}  
A key ingredient in this paper (and also in the forthcoming paper~\cite{DYZ2})
will be a pair $(\psi,\psi^*)$, consisting of 
a wave function~$\psi$ and the dual wave function~$\psi^*$ associated with~$\psi$. 
We explain this briefly here and in detail in Section~\ref{section3}.

We start with the time-independent case. 
Let $f(x)$ be an element of~$V$, and $L$ the linear Schr\"odinger operator $\p_x^2+ 2f(x)$. 
By a {\it wave function} of~$f$ we will mean 
an element $\psi=\psi(z,x)$ 
in the module $\wt V \bllb z^{-1}\brrb\, e^{xz}$ satisfying the equation $L(\psi) = z^2 \psi$ 
of the form $\psi = \bigl(1+ \phi_1(x)/z+\phi_2(x)/z^2 + \cdots\bigr) \,e^{xz}$, where 
$\wt V$ is any ring 
with $V\subseteq \p_x \bigl(\wt V \bigr) \subseteq \wt V$. 
The {\it dual wave function}\footnote{A dual wave function satisfies $L^*\psi^* = \lambda \, \psi^*$, where $L^*$ 
is the formal adjoint operator of~$L$. For the KdV hierarchy, 
$L^*=L$ and a dual wave function~$\psi^*$ is also a wave function.} $\psi^*$ of~$f$ associated with~$\psi$ is then defined as 
the unique element in~$\wt V\bllb z^{-1}\brrb \, e^{-xz}$ satisfying  
$L(\psi^*) = z^2 \psi^*$ of the form 
$\psi^* = \bigl(1+ \phi_1^*(x)/z+\phi_2^*(x)/z^2 + \cdots\bigr) \,e^{-xz}$ 
for which $\p_x^i(\psi) \, \psi^*$  has 
residue~0 at $z=\infty$ for all $i\geq 0$. 
The pair $(\psi,\psi^*)$ consisting of a wave function~$\psi$ of~$f$ and the unique dual wave function~$\psi^*$ associated with~$\psi$ 
is called {\it a pair of wave functions} of~$f$. 
Given~$f$, the wave function~$\psi$~of~$f$ is   
unique up to multiplication by an arbitrary element in $1+z^{-1} \CC \bllb z^{-1}\brrb$, but $\psi \psi^*$ 
is unique and coincides with~$b(z^2,x)$, the upper right entry of the matrix $\R(z^2,x)$ defined above, as is proved in Lemma~\ref{twoid}.

We now proceed to the time-dependent case. Let $L=\p^2 + 2u_0$ be a linear operator on~$\A$, called the Lax operator for the KdV hierarchy, and introduce a sequence of differential 
operators $A_k$ defined by $A_k = \frac{1}{(2k+1)!!} \bigl(L^{\frac{2k+1}2}\bigr)_+$~($k\geq 0$). 
Here $(\cdot)_+$ means taking the differential part of a pseudo-differential operator (\cite{Dickey}). 
Let $u=u(\bt)$ be an arbitrary solution to the KdV hierarchy~\eqref{KdVhk}. By a {\it wave function} of~$u$ we will mean 
an element $\psi=\psi(z,\bt)$ 
in $\wt V\llm \bt_{>0}\rrm \bllb z^{-1}\brrb \, e^{\sum_{k= 0}^\infty t_k \, z^{2k+1} /(2k+1)!!}$ 
of the form
$\psi =  \bigl(1+ \phi_1(\bt)/z+ \phi_2(\bt)/z^2 +\cdots \bigr) \, e^{\sum_{k= 0}^\infty t_k \, z^{2k+1} /(2k+1)!!}$  
satisfying the equations 
\beq \label{definingpsi} 
L (\psi) \= z^2 \psi \,, \qquad  \p_{t_k} (\psi) \= A_k   (\psi)\,. 
\eeq 
Define the {\it dual wave function}~$\psi^*=\psi^*(z,\bt)$ of~$u$ associated with~$\psi$ as the unique 
element in $\wt V\llm \bt_{>0}\rrm \bllb z^{-1}\brrb \, e^{-\sum_{k= 0}^\infty t_k \, z^{2k+1} / (2k+1)!!}$ 
of the form 
$\psi^* = \bigl(1+\phi_1^*(\bt)/z+ \phi_2^*(\bt)/z^2 +\cdots \bigr) \, e^{-\sum_{k= 0}^\infty t_k \, z^{2k+1} /(2k+1)!!}$ 
satisfying
\beq\label{definngspsistar}
L(\psi^*) \= z^2 \psi^*\,, \qquad  -\p_{t_k} (\psi^*) \= A_k  (\psi^*) 
\eeq 
and that $\p_x^i(\psi(z,\bt)) \, \psi^*(z,\bt)$ has residue~0 at $z=\infty$ for all $i\geq 0$. 
We say that a wave function~$\psi$ of~$u$ and the dual wave function~$\psi^*$ of~$u$ associated with~$\psi$ 
form a {\it pair of wave functions} of~$u$. 
The existence of a pair of wave functions of~$u$, which is known, will be shown in Section~\ref{section3}.
Given a solution~$u=u(\bt)$, the wave function~$\psi$~of~$u$ is   
unique only up to multiplication by an arbitrary element in $1+z^{-1} \CC \bllb z^{-1}\brrb$, but the
product $\psi \psi^*$ is unique and coincides with~$b(z^2,\bt)$, the upper right entry of $R(z^2,\bt)$, 
where $\psi^*$ is the dual wave function of~$u$ associated with~$\psi$.  

Let $u=u(\bt)$ be an arbitrary solution of~\eqref{KdVhk} in $V\llm \bt_{>0}\rrm $. 
The basic matrix resolvent~$R(z^2,\bt)$ of~$u$ can be expressed in terms of  
a pair of wave functions of~$u$ (Lemma~\ref{factorize1}, Section~\ref{section3}). 
This enables us to give a new formula for the generating series of $n$-point correlation functions $\Omega_{p_1, \dots, p_n} (\bt)$ of~$u$. 
Let $(\psi(z,\bt),\psi^*(z,\bt))$ be a pair of wave functions of~$u$. Define 
\beq\label{Dzwtdefinition}
D(z,w,\bt) \:= \frac{\psi(z,\bt) \, \psi_x^*(w,\bt) \, - \, \psi^*(w,\bt) \, \psi_x(z,\bt)}{w^2\,-\,z^2} \,.
\eeq

\begin{theorem}\label{wavethm} 
For any fixed $n\geq 2$ the following formula 
holds true:
\begin{align}
& \sum_{p_1,\dots,p_n} \Omega_{p_1, \dots, p_n} (\bt) \, \prod_{j=1}^n \frac{(2p_j+1)!!}{z_j^{2p_j+2}}  \= 
- \sum_{\sigma\in S_n/C_n} 
\prod_{i=1}^n D\bigl(z_{\sigma(i)}, z_{\sigma(i+1)}, \bt \bigr)  
\,- \, \frac{\delta_{n2}}{(z_1-z_2)^2} \,. \label{newformula}
\end{align}
\end{theorem}
We note that under the gauge freedom of the second type $\psi(z,\bt)\mapsto g(z) \psi(z,\bt) = \wt \psi(z,\bt)$, the dual wave function $\psi^*$ associated with~$\psi$ 
is mapped to $\psi^*(z,x)/g(z) = \wt \psi^*(z,x)$. Therefore we have
$$\wt D(z,w,\bt) \:= \frac{ \wt \psi(z,\bt) \, \wt \psi_x^*(w,\bt) \, - \, \wt \psi^*(w,\bt) \, \wt \psi_x(z,\bt)}{w^2\,-\,z^2} \= \frac{g(z)}{g(w)} \, D(z,w,\bt) \,.  $$
However, products of the factors of the form $g(z)/g(w)$ cancel in each sum of the right hand side of the formula~\eqref{newformula} and 
therefore remain unchanged under the gauge freedom of the second type.  This argument agrees with the fact that the $n$-point correlation functions 
$\Omega_{p_1,\dots,p_n}(\bt)$ of~$u$ (for $n\geq 2$) are the evaluations of certain elements of~$\A$ with $u_i$ replaced by $\p_x^i(u)$. 

For an element $f(x)\in V$ and a solution $u(\bt)\in V\llm \bt_{>0}\rrm $ to the KdV hierarchy related via the 1-1 
correspondence~\eqref{11correspondence}, we have $\psi(z,x,\bdzero)\equiv\psi(z,x)$. For a given~$f$, 
one can find an explicit recursion for solving~$\psi(z,x)$ (see Section~\ref{section3}). 
However, we do not know an efficient way of solving~$\psi(z,\bt)$. This is 
because the recursive procedure of solving $\psi(z,\bt)$ requires the knowledge of~$u(\bt)$, but this requires first solving the KdV 
hierarchy.  However, due to~\eqref{newformula}, the knowledge of~$\psi(z,x)$ gives rise to a construction of 
the logarithm of the tau-function~$\tau(\bt)$ of~$u$ (since $u=\p_x^2(\log \tau)$, \eqref{newformula} also constructs $u$) in a simple way. 
More precisely, if we specialize~\eqref{newformula} to $\bt=(x,\bdzero)$, then it gives a formula
$\Omega_{p_1,\dots,p_n}(x)$ for all $n\geq 2$ and $p_i\geq 0$ and hence, in view of~\eqref{Taylorlogtau}, gives 
the entire Taylor series of~$\log \tau(\bt)$ (which, we recall, is only defined up to the addition of 
a linear term). This is important because in our concrete computations for Examples~1--3 of 
Section~\ref{section11}, we can compute $\psi(z,x)$ and $\psi^*(z,x)$, and hence the function $D(z,w,x)$, explicitly.

\begin{remark}
For $V=\CC[[x]]$, one could use alternatively  
the Sato Grassmannian approach~\cite{BDY1,DKJM,S, Zhou2} to prove the identity~\eqref{newformula}. 
It would be interesting to investigate if the identity~\eqref{newformula} 
could also be proved, say for $V=\CC\llb x\rrb$, by using the 
approach of Berg\`ere and Eynard with the employment of 
topological recursion (loop equation) \cite{AvM,BEM1,BE,CEO,DVV,EO,Marchal,Zhou-TR}, or by using matrix models together 
with appropriate Riemann--Hilbert problems 
(isomonodromic deformations) \cite{BC, BeR0, BeR, Deift, DVV, DuY1, DuY3, HZ, Kont, M, Witten}, or by using OPE from 
appropriate vertex algebras \cite{BDM, DKJM, FZ, Zhou2};  
we expect that at least for the initial data $f(x)\in V$ having  
the {\it bispectral} property defined by Duistermaat and Gr\"unbaum~\cite{DG} 
(see also the $M$-bispectrality given in Section~\ref{bispsection})  
this might be possible. We also note that the matrix resolvent method and some of the above-mentioned methods extend to new situations 
(see e.g. \cite{BDY2,BDY3,DuY1,DuY2,DuY3,DYZ}, \cite{Alexandrov2}, 
 \cite{BEM1,BEM2,BBE, Marchal} and \cite{BeR0,BeR,BeR2,GGR}).
\end{remark}

We now formulate a theorem that is equivalent to Theorem~\ref{wavethm} but involves 
a function $K(z,w,\bt)$ which, unlike $D(z,w,\bt)$, depends only on the solution~$u(\bt)$ 
(i.e. it is independent of the choice of a pair of wave functions of~$u(\bt)$).  The function $K(z,w,\bt)$ is defined by
\begin{align} 
K(z,w,\bt) \:= 
\frac{ b(z^2,\bt) \, b_x(w^2,\bt)  -  b(w^2,\bt) \, b_x(z^2,\bt) }{2 \, (w^2-z^2)}  \,-\, 
\frac{ w\, b(z^2,\bt) + z\, b(w^2,\bt) }{w^2-z^2} \, .  \label{defbigK}
\end{align}
\begin{theorem}\label{thmmainabs} 
For any fixed $n\geq 2$ the following formula 
holds true:
\begin{align}
&  \sum_{p_1,\dots,p_n} \Omega_{p_1, \dots, p_n} (\bt) \, \prod_{j=1}^n \frac{(2p_j+1)!!}{z_j^{2p_j+2}}  \= 
- \frac1{\prod_{i=1}^n b\bigl(z_i^2,\bt\bigr)} \, \sum_{\sigma\in S_n/C_n} 
\prod_{i=1}^n K\bigl(z_{\sigma(i)}, z_{\sigma(i+1)}, \bt \bigr)  \,-\, \frac{\delta_{n2}}{(z_1-z_2)^2} \,. \label{newformulaabs}
\end{align}
\end{theorem}
\begin{remark} 
Both the formula~\eqref{newformula} and the formula~\eqref{newformulaabs} are generalized 
to the Toda lattice hierarchy in~\cite{Y2}. 
For example, the following theorem is proved in~\cite{Y2}.

\noindent {\bf Theorem}~(\cite{Y2}).
{\it Fix $k\geq 2$ an integer. 
For an arbitrary solution $(v,w)$ to the Toda lattice hierarchy associated with the 
Lax operator $L_{\rm Toda}=\Lambda+v+w\Lambda^{-1}$, 
let $\psi_1(\lambda,n,\bt),\psi_2(\lambda,n,\bt)$ be a pair of wave functions of the solution~$(v,w)$. 
Here $\Lambda$ denotes the shift operator $\Lambda:f(n)\mapsto f(n+1)$. Define
$$
D(\lambda,\mu,n,\bt) \= \frac{\psi_1(\lambda,n,\bt) \, \psi_2(\mu,n-1,\bt) \,-\, \psi_1(\lambda,n-1,\bt) \,\psi_2(\mu,n,\bt)}{\lambda-\mu} \,.
$$
Then the following formula holds true:
\begin{align}
& \sum_{i_1,\dots,i_k\geq 0} 
\frac{\Omega_{i_1,\dots,i_k}^{\rm Toda}(n,\bt)}{\lambda_1^{i_1+2} \cdots \lambda_k^{i_k+2}}  \=  
(-1)^{k-1} \frac{e^{kq(n-1,\bt)} }{\prod_{j=1}^k \lambda_j} \sum_{\sigma\in S_k/C_k} \prod_{j=1}^k D\bigl(\lambda_{\sigma(j)},\lambda_{\sigma(j+1)},n,\bt\bigr) \,-\, \frac{\delta_{k,2}}{(\lambda_1-\lambda_2)^2}    \,, 
\end{align}
where $\Omega_{i_1,\dots,i_k}^{\rm Toda}(n,\bt)$ 
denote the $k$-point correlation functions of $(v,w)$, 
and the $q(n,\bt)$ is defined via}
$$
w(n,\bt) \= e^{q(n-1,\bt)-q(n,\bt)} \,, \quad  \frac{\p q(n,\bt)}{\p t_i} = - S_i(n,\bt),~i\geq0\,.
$$
\noindent 
See~\cite{Y2} for the precise definitions of the functions $S_i(n,\bt)$, $\Omega_{i_1,\dots,i_k}^{\rm Toda}(n,\bt)$, and of 
a pair of wave functions for the Toda lattice hierarchy.
\end{remark}

We observe that all the following three functions 
$$
K(z,w,\bt) \,-\, \frac{b\bigl(w^2,\bt\bigr)}{z-w} \,, \quad K(z,w,\bt) \,-\, \frac{b\bigl(z^2,\bt\bigr)}{z-w} \,, \quad K(z,w,\bt) \,-\, \frac12 \frac{b\bigl(z^2,\bt\bigr)+b\bigl(w^2,\bt\bigr)}{z-w} 
$$
belong to the ring $V \llm \bt_{>0}\rrm \bllm z^{-1},w^{-1} \brrm$.
We also remark that if we define $K$ as the right hand side of~\eqref{defbigK} with $b(z^2,\bt)$ and $b(w^2,\bt)$ 
replaced by $b(z^2)$ and $b(w^2)$, respectively, then it easily follows that
the three functions $K - b(z^2)/(z-w) $, $K - b(w^2)/(z-w) $ and 
$K - \frac12 \bigl(b(z^2) + b(w^2) \bigr)/(z-w)$ all belong to $\A  \bllm z^{-1},w^{-1} \brrm $. 
Moreover, the identity~\eqref{newformulaabs} is then also true in
the {\bf abstract sense} (similarly to Proposition~\ref{taulemma}). 
We have for example
\begin{align}
& K \,-\, \frac12 \frac{b\bigl(z^2\bigr)+b\bigl(w^2\bigr)}{z-w}  \= \frac{u_0}2 \biggl(\frac{1}{zw^2} - \frac{1}{z^2w}\biggr) \,-\, \frac{u_1}{2} \frac1{ z^2w^2} 
\+ \frac{6u_0^2+u_2}{8} \biggl(\frac1{zw^4} - \frac{1}{z^2w^3} + \frac{1}{z^3w^2} - \frac{1}{z^4w}\biggr) \nn\\
& \qquad \qquad \qquad \qquad \qquad \qquad \qquad  - \frac{12 u_0u_1+u_3}{8} \biggl( \frac1{z^2w^4} + \frac1{z^4w^2}\biggr) \+ \cdots \,. \nn 
\end{align}
\begin{prop}\label{KDrelation} Let $\psi(z,\bt)$ be any wave function of~$u(\bt)$.
The functions $K$ and~$D$ satisfy the relation
\begin{align}
K(z,w,\bt) \= b\bigl(z^2,\bt\bigr) \,\frac{\psi(w,\bt)}{\psi(z,\bt)}  \, D(z,w,\bt) 
 \= b\bigl(w^2,\bt\bigr) \, \frac{\psi^*(z,\bt)}{\psi^*(w,\bt)} \, D(z,w,\bt) \= \psi^*(z,\bt) \, \psi(w,\bt)  \, D(z,w,\bt)  \,.  \label{KDequation1} 
\end{align}
\end{prop}
The proofs of Proposition~\ref{KDrelation} and Theorem~\ref{thmmainabs} will be given in Section~\ref{section3}.

Let us now consider the case $V=\CC\llm x\rrm$, namely,     
consider arbitrary solutions of the KdV hierarchy in the ring $\CC\llm t_0, t_1,t_2,\dots \rrm$, where we recall that $t_0=x$.  
Define $D(z,w)$ and $K(z,w)$ as follows:
\beq\label{DandK}
D(z,w)\:=D(z,w,\bdzero) \= \frac{\psi(z,\bdzero) \, \psi_x^*(w,\bdzero) \, - \, \psi^*(w,\bdzero) \, \psi_x(z,\bdzero)}{w^2\,-\,z^2} \,, 
\qquad  K(z,w)\:=K(z,w,\bdzero)  \,. 
\eeq 
Then Proposition~\ref{KDrelation} immediately implies the following corollary. 
\begin{cor} \label{corofnew} 
For $V=\CC\llm x\rrm$, 
we have 
\beq\label{eq24}
D(z,w) \,-\, \frac1{z-w} \in \CC  \bllm z^{-1},w^{-1} \brrm  \,,  
\eeq 
\beq \label{eq25} 
K(z,w) \,-\, \frac{b(z^2,{\bf 0})}{z-w} \,, ~ K(z,w) \,-\, \frac{b(w^2,{\bf 0})}{z-w} \,, 
~ K(z,w) \,-\, \frac12 \frac{b(z^2,{\bf 0})+b(w^2,{\bf 0})}{z-w} \;  \in  \CC  \bllm z^{-1},w^{-1} \brrm \, .
\eeq
\end{cor}
\noindent Note that for this case, pairs of wave functions 
of the solution correspond to certain points 
of the Sato Grassmannian. In such correspondences, 
the coefficients of $D(z,w)-1/(z-w)$ coincide with the affine coordinates~\cite{EH} of the point of the Sato Grassmannian.
Based on this observation we find that for~$V=\CC\llm x\rrm$ the identity~\eqref{newformula} evaluated at $\bt=\bdzero$ 
is equivalent to a formula obtained by J.~Zhou~\cite{Zhou2} from the fermionic picture (see the Theorem~5.3 of~\cite{Zhou2}). 
The affine coordinates can alternatively be used to compute the Pl\"ucker coordinates 
in the expansion of tau-function (as opposed to its logarithm) (cf. \cite{BY,EH,S,Zhou0,Zhou2}).
Geometric meaning (in Sato Grassmannian) of the coefficients of the power series in~\eqref{eq25} 
remains an interesting open question,
which we plan to do in a subsequent publication.

\subsection{Applications}  \label{Introapp}

We now discuss each of the three examples from Section~\ref{section11} in turn.

\medskip

\noindent {\bf  Example 1: Intersection numbers of~$\psi$-classes.} 
An explicit formula for the generating series of~\eqref{wkcor}
in terms of the basic matrix resolvent was derived in~\cite{BDY1}; see also~\eqref{npointwkold} and~\eqref{MWK}. 
A particular pair of wave functions $(\psi(z,x),\psi^*(z,x))$ of $f=x$ were given in~\cite{BDY1} (cf. also~\cite{Buryak1,KS}):
\begin{align}
& \psi(z,x) \= \frac{\sqrt{z}}{(z^2-2x)^{\frac14}} \, e^{\frac13z^3-\frac13 (z^2-2x)^{\frac32}} \sum_{k\geq 0} \frac{(-1)^k}{288^k} \frac{(6k)!}{(3k)! \, (2k)!} \, (z^2-2x)^{-\frac{3k}2} \,, \label{psiwk1} \\  
& \psi^*(z,x) \= \frac{\sqrt{z}}{(z^2-2x)^{\frac14}}  \, e^{-\frac{z^3}3 + \frac13 (z^2-2x)^{\frac32}}  \, \sum_{k\geq 0} \frac{1}{288^k} \frac{(6k)!}{(3k)! \, (2k)!} \, (z^2-2x)^{-\frac{3k}2} \,.  \label{psiwk2}
\end{align}
They can be viewed as the asymptotic expansion of the following two analytic functions, respectively, as $z\rightarrow \infty$ within appropriate sectors~\cite{BDY1,BeY}:
$$
\sqrt{2\pi z} \, e^{\frac{z^3}3} \, 2^{\frac13} \, {\rm Ai}(\xi)\,, \qquad  e^{\frac{\pi i}6} \, \sqrt{2\pi z} \, e^{-\frac{z^3}3}  \, 2^{\frac13} \, {\rm A_i}(\omega \xi) \,, 
$$
where $\xi=2^{-\frac23} \, (z^2-2x)$, $\omega=e^{2\pi i/3}$. 
Restricting the formal series~\eqref{psiwk1}--\eqref{psiwk2} to~$x=0$ we obtain 
\begin{align}
& \psi(z,0) \= \sum_{k=0}^\infty \frac{(-1)^k}{288^k} \frac{(6k)!}{(3k)! \, (2k)!} \, z^{-3k} \,,  \qquad \psi^*(z,0) \= \psi(-z,0) \,, \\
& \psi_x(z,0) \= \sum_{k=0}^\infty \frac{1+6k}{1-6k} \frac{(-1)^k}{288^k} \frac{(6k)!}{(3k)! \, (2k)!} \, z^{-3k+1} \,,    \qquad 
 \psi_x^*(z,0) \= \psi_x(-z,0) \,.
\end{align}
These series appeared in the Faber--Zagier formula~\cite{F} on relations in the tautological ring of the moduli space of curves. 
Now consider the function 
$D(z,w) = D(z,w, \bdzero)$.
According to~\eqref{eq24} we know that  
this function has a pole only at $z=w$.
The first few terms for $D(z,w)-1/(z-w)$ are given by
\begin{align}
& D(z,w) \,-\, \frac1{z-w} \=  \frac{5}{24} z^{-1} w^{-3} \,-\, \frac7{24} z^{-2} w^{-2} \+ \frac5{24} z^{-3} w^{-1} \,-\, \frac{455}{1152} z^{-2} w^{-5}  \nn\\
& \qquad \qquad \qquad \qquad \qquad \+ \frac{385}{1152} z^{-3} w^{-4} \,-\, \frac{385}{1152} z^{-4} w^{-3}  \+ \frac{455}{1152} z^{-5}w^{-2} \+ \cdots \,.  \label{Dcoeff}
\end{align}

Theorem~\ref{wavethm} then implies the following slightly simplified version of a formula of J.~Zhou:

\medskip

\noindent {\bf Theorem (Zhou \cite{Zhou2})}.
{\it For $n\geq 2$, the generating series of the $n$-point intersection numbers of $\psi$-classes has the expression}
\begin{align}
& \sum_{g\geq 0}  \sum_{p_1,\dots,p_n\geq 0}  \int_{\overline{\M}_{g,n}} \psi_1^{p_1} \cdots  \psi_n^{p_n} \, 
\prod_{j=1}^k \frac{(2p_j+1)!!}{z_j^{2p_j+2}} \= - \sum_{\sigma\in S_n/C_n} \prod_{i=1}^n D\bigl(z_{\sigma(i)}, z_{\sigma(i+1)} \bigr)  
\,- \, \frac{\delta_{n2}}{(z_1-z_2)^2} \,. \label{npointwkzhou}
\end{align}

We observe that an explicit formula for the coefficients in~\eqref{Dcoeff} was 
obtained by J.~Zhou~\cite{Zhou0} by solving Virasoro constraints using the {\it fermionic} method. Later 
the explicit formula was re-proved by F.~Balogh and one of the authors of the present paper~\cite{BY} using the Sato 
Grassmannian approach~\cite{S}.

The Laurent series of the corresponding function~$K(z,w)$ begins
\beq
K(z,w) \,-\, \frac12 \frac{b(z^2,\bdzero)+b(w^2,\bdzero)}{z-w}  \= -\frac12 \frac1{z^2w^2} \, + \,
\frac5{16} \biggl(\frac1{zw^6} - \frac1{z^2w^5} +\frac1{z^3w^4} - \frac1{z^4w^3} +\frac1{z^5w^2} -\frac1{z^6w} \biggr) \+ \cdots.  
\eeq

\medskip

\noindent {\bf  Example 2: Theta-classes.}
For the BGW-Norbury solution $u_{\Theta(1/8)}$ introduced in Example~\ref{example2} (with $C=1/8$), denote by 
$\oteight_{p_1,\dots,p_n}(\bt)$ the $n$-point correlation functions of~$u_{\Theta(1/8)}$.
According to Norbury~\cite{N} $\oteight_{p_1,\dots,p_n}(\bt)$ for $n\geq 2$ are related to the $\Theta$-class intersection numbers by 
\beq 
\oteight_{p_1, \dots,p_n}  (\bdzero) \= \int_{\overline{\mathcal{M}}_{g,n}} \Theta_{g,n} \, \psi_1^{p_1} \cdots \psi_n^{p_n} \qquad (g=p_1+\cdots+p_n+1)\,.  
\label{Norburyci} 
\eeq
(The expression on the right vanishes for $g\neq p_1+\dots+p_n+1$.)  
For any $n\geq 1$, denote by $F(\lambda_1,\dots,\lambda_n)$ the following generating series of the $n$-point $\Theta$-class intersection numbers
\beq 
F(\lambda_1,\dots,\lambda_n) \= \sum_{p_1,\dots,p_n} \prod_{j=1}^n \frac{(2p_j+1)!!}{\lambda_j^{p_j+1}}   
\int_{\overline{\mathcal{M}}_{1+p_1+\cdots+p_n,n}} \Theta_{1+p_1+\cdots+p_n,n} \, \psi_1^{p_1} \cdots \psi_n^{p_n}  \,. 
\eeq
Using the relationship~\eqref{Norburyci} and Corollary~\ref{cor2} we will prove:
\begin{theorem} \label{BGWnpoint} 
For fixed $n\geq 2$, the generating series $F(\lambda_1,\dots,\lambda_n)$ has the expression
\begin{align}
F(\lambda_1,\dots,\lambda_n)  \= -\sum_{\sigma\in S_n/C_n} \frac{\tr \, \bigl( M (\lambda_{\sigma(1)})\cdots M(\lambda_{\sigma(n)})\bigr)}
{\prod_{i=1}^n \bigl(\lambda_{\sigma(i+1)}-\lambda_{\sigma(i)}\bigr)} 
\,-\, \delta_{n,2} \frac{\lambda_1+\lambda_2}{(\lambda_1-\lambda_2)^2} \, ,  \label{state2}
\end{align}
where $M(\lambda)= \begin{pmatrix} 0 & 0\\ \lambda & 0\end{pmatrix} \+ \displaystyle\sum_{k\geq 0} \, \biggl[\dfrac{(2k-1)!!}{2^k}\biggr]^3  \,
\biggl (\begin{matrix} k  & 1\\
-  \frac{8k^3+12 k^2 + 4 k + 1}{8 (k+1)} & -k 
\end{matrix}  \biggr) \,  \dfrac{\lambda^{-k}}{k!} \, $.
\end{theorem}
The string type equation (see in Section~\ref{section4}) relates the 2-point $\Theta$-class intersection numbers to 1-point intersection numbers, from which we deduce:
\begin{cor}\label{1pointbgw}
The 1-point $\Theta$-class intersection numbers have the expression
\beq\label{state1}
\int_{\overline{\mathcal{M}}_{g,1}} \Theta_{g,1} \, \psi_1^{g-1}  \= \frac{(2g-1)!!^2}{8^g \, g! \, (2g-1)}  \,, \quad \forall\,g\geq 1\,.
\eeq
\end{cor}
For a general~$C$ in Example~\ref{example2}, let~$\tauC(\bt)$ be the tau-function of~$\utc$, 
$R(\lambda,\bt)$ the basic matrix resolvent of~$\utc$, and $\otc_{p_1,\dots,p_n}(\bt)$ the $n$-point correlation functions of~$\utc$.
We will call $\otc_{ p_1, \dots, p_n} (\bdzero)$ 
 the $n$-point generalized BGW correlators. 
For $n\geq2$, $\otc_{ p_1, \dots, p_n} (\bdzero)$ are 
 uniquely determined by~$\utc$ and so by~$f$, where $f=\frac{C}{ (x-1)^2}$. To determine
$\otc_{ p} (\bdzero)$, we impose a string-type equation for~$\tauC$
\beq\label{scaling}
\sum_{i\geq 0} (1+2i) \, t_i \, \frac{\p \log \tauC}{\p t_i} \+ C \= \frac{\p \log \tauC}{\p t_0} \,. 
\eeq
Since $\sum_{i\geq 0} (1+2i) \, t_i \, \frac{\p }{\p t_i}$ gives the infinitesimal scaling symmetry of the KdV hierarchy 
and since eq.~\eqref{scaling} implies that $\utc(x,\bdzero)=C/(x-1)^2$, 
the tau-function~$\tauC$ of~$\utc$ satisfying the additional constraint~\eqref{scaling} exists. 
(From this way, $\tauC$ is determined up to a multiplicative constant.)  
A matrix model for~$\tauC$ was obtained by A.~Alexandrov~\cite{A}. 
The function $\tauC$ is called the {\it generalized BGW tau-function}. 
As before denote $\R(\lambda, x)=R(\lambda,x,\bdzero)$, $b(\lambda, x)= \R(\lambda, x)_{12}$. 

\begin{theorem}  \label{RCC} For any fixed $C\in\CC$, 
the formal power series $b(\lambda, x)$ of~$\lambda^{-1}$ only depends on the variable $\zeta:=\lambda\, (x-1)^2$. 
Explicitly,  
the matrix-valued formal series $\R=\R(\lambda,x)$ has the explicit expression
\begin{align}
\R \= \begin{pmatrix}
-\frac {C\, \lambda^{\frac12}} {\z^{\frac32}}  G_{\frac32}(\z) &  G_{\frac12}(\z) \\
 \frac{\lambda}\z \, \Bigl((\z-2C) \, G_{\frac12}(\z)
- \frac{3C}\z \, G_{\frac32}(\z) - \frac{6C(C+1)}{\z^2} G_{\frac52}(\z)  \Bigr)  & \frac {C\, \lambda^{\frac12}} {\z^{\frac32}}  G_{\frac32}(\z) \\
\end{pmatrix} \,, \label{RTheta}
\end{align}
where $\z^{3/2}:=\lambda^{3/2}\, (x-1)^3$ and 
we have set
\beq\label{defG}
\Delta \:= 1- 8 \, C \,, \quad G_\alpha(\z) \:= {}_{3} F_0 \biggl(\alpha, \, \alpha+\frac{\sqrt{\Delta}}2, \, \alpha-\frac{\sqrt{\Delta}}2; \, ; \frac1\z \biggr)\,.
\eeq
\end{theorem}
\noindent The proof is in Section~\ref{section4}. It is interesting to mention that if $-C$ is a 
triangular number, i.e. $C=-\frac12 p(p+1)$~$(p\geq 0)$ then $b(\lambda,x)$ truncates to a polynomial of $\lambda^{-1}$.

\begin{cor}\label{1pointgbgw}
The 1-point generalized BGW correlators have the explicit expression
\beq\label{1ptC}
\otc_{p} (\bdzero) \= \frac{1}{(p+1)! \, (1+2p)} \, \prod_{i=0}^{p} \biggl(C+\frac{i(i+1)}{2}\biggr) \,, \quad j\geq 0\,.
\eeq
\end{cor}

We proceed to applying Theorem~\ref{wavethm} to the computation of generalized BGW correlators. 
\begin{theorem}\label{dbessel} Let $\alpha= \sqrt{\frac14 - 2C}$, and for $k\ge0$ define 
$a_k(\alpha)\in\mathbb Q[\alpha^2]=\mathbb Q[C]$ by
\[ 
a_k(\alpha) \= \frac{\bigl(\alpha-k+\frac12\bigr)_{2k}}{2^k\, k!} 
 \= \frac{(-1)^k}{k!}\,\prod_{j=1}^k\biggl(C+\binom j2\biggr) \= (-1)^k(2k-1)\,\otc_{k-1} (\bdzero) \,.  \]
Then the two functions
\beq
\psi(z,x) \=  \sum_{k\geq 0}  \frac{a_k(\alpha)}{z^k(1-x)^k}  \, e^{zx}  \,, 
\qquad  \psi^*(z,x) \= \psi(-z,x)
\label{gbgwbessel}
\eeq
form a pair of wave functions of $f=C/(x-1)^2$.  Moreover, if $\DzwC=D(z,w)$ is the corresponding 
function as defined in~\eqref{DandK}, then for any $n\geq 2$ we have the generating function
\begin{align}
& \sum_{p_1,\dots,p_n} \otc_{p_1, \dots, p_n} (\bdzero) \, \prod_{j=1}^n \frac{(2p_j+1)!!}{z_j^{2p_j+2}}  \= 
- \sum_{\sigma\in S_n/C_n} 
\prod_{i=1}^n \DC\bigl(z_{\sigma(i)}, z_{\sigma(i+1)} \bigr)  
\,- \, \frac{\delta_{n2}}{(z_1-z_2)^2} \,. \label{newformulaforgbgw}
\end{align}
\end{theorem}
According to equation~\eqref{eq24} we know that $D(z,w)$ has an expansion of the form
\beq  \DzwC \= \frac1{z-w} \+ \sum_{m,\,n\geq 0} \frac{A_{mn}(\alpha)}{z^m\,(-w)^n}  \eeq
with coefficients $A_{mn}(\alpha)\in\Bbb Q[\alpha]$.
These coefficients are given explicitly by the following proposition.
\begin{prop}\label{identities}  
The coefficients $A_{mn}$ can be given either by the recursion and boundary conditions
\beq \label{amnrec} A_{m,n+1}(\alpha) \;-\;A_{m+1,n}(\alpha) \= \frac{m-n}{m+n}\,a_m(\alpha)\, a_n(\alpha)\,, \qquad
   A_{m,0}(\alpha)=A_{0,n}(\alpha)=0 \eeq
or else explicitly by any of the three formulas
\begin{align} \label{amnsimple} 
A_{mn}(\alpha) &\= \sum_{r\ge m,\;s\ge0\atop r+s=m+n-1}\frac{r-s}{r+s}\,a_r(\alpha)\,a_s(\alpha)  
\= \sum_{r\ge n,\;s\ge 0\atop r+s=m+n-1}\frac{r-s}{r+s}\,a_r(\alpha)\,a_s(\alpha)  \\ 
\label{amnclosed} &\= 2\,\frac{m!\,n!\,a_m(\alpha) \, a_n(\alpha)}{(m+n-1) \, (m+n-1)!}  \,
  \sum_{k\in \ZZ} \,(-1)^k \,\frac{ \binom{m+n-1}{m+k} \binom{m+n-1}{n+k} }{\alpha-k-\frac12} \,. 
 \end{align}
\end{prop}
\noindent Observe that the summation in~\eqref{amnclosed} is finite (since the $k_{\rm th}$ term vanishes unless $|k+1/2|<m,n$)
and is even in~$\alpha$ (as one sees by sending $\alpha$ to $-\alpha$ and $k$ to $-k-1$), and that either 
formula~\eqref{amnsimple} or~\eqref{amnclosed} implies that $A_{mn}(\alpha)/a_m(\alpha)$ for fixed $m$~and~$n$ 
is a polynomial of degree~$n-1$ in~$\alpha^2$.

The Laurent series of the corresponding function $K(z,w)$ (defined in~\eqref{DandK}) begins
\begin{align}
& K(z,w) \,-\, \frac12 \frac{b\bigl(z^2,\bdzero\bigr)+b\bigl(w^2,\bdzero\bigr)}{z-w}\nn\\
& \quad \= 
\frac{C}{2} \biggl(\frac{1}{z w^2} \,-\, \frac1{ z^2 w}\biggr) \,-\,  \frac{C}{z^2w^2}  \+ 
\frac{3C(1+C)}{4} \biggl(\frac1{ z w^4} - \frac1{z^2 w^3} + \frac1{z^3w^2} - \frac1{z^4w}\biggr) \+ \cdots \,. 
\end{align}

We also note that, similarly as above, 
the pair of wave functions $\psi,\psi^*$ of the above theorem can be viewed as the 
asymptotic expansions of the following two analytic functions, respectively, as $z$ goes 
to~$\infty$ within a certain sector of the complex $z$-plane:
\beq
\sqrt{\frac{2\xi}\pi} \, e^z \, K_\alpha(\xi) \,, \qquad \sqrt{2\pi \xi } \, e^{-z} \, I_\alpha(\xi) \,, 
\eeq
where $I_\alpha$ and $K_\alpha$ denote the modified Bessel functions and $\xi=z \, (1-x)$. Finally we note that 
the function $\DzwC$ in terms of these analytic functions reads
\beq
\DzwC \= 2 \, e^{z-w} \sqrt{zw} \; \frac{K_\alpha(z) \, w \, I_\alpha'(w) - I_\alpha(w) \, z\, K_\alpha'(z) }{z^2-w^2} \,, 
\eeq
which can be recognized as a Bessel kernel (cf.~\cite{BT,TW}). 

\medskip

\noindent {\bf  Example 3: Lam\'e solutions.} 
Let $\tau_{{\rm elliptic}}$ be the tau-function of the Lam\'e solution $u_{{\rm elliptic}}(\bt)$ 
introduced above, which as we will see gives an elliptic generalization of~$\tau_{\Theta(C)}$.
Denote by $R(\lambda,\bt)$ the basic matrix resolvent of~$u_{\rm elliptic}$, and denote 
$\R(\lambda, x)=R(\lambda,x,\bdzero)$ and $b(\lambda, x)= \R(\lambda, x)_{12}$. 
If we restrict to the particular case when $-C$ is a triangular number, i.e. $C=-\frac{p(p+1)}2$ for some integer $p\geq 0$, 
then $u_{{\rm elliptic}}(\bt)$ is a special {\it finite-gap} solution~\cite{DN1,Du0} of the KdV hierarchy. Denote by 
\beq
y^2 \= S_p(\lambda)
\eeq
the corresponding {\it spectral curve} \cite{DN1,Du0}, where $S_p(\lambda)$ is a polynomial of $\lambda$ of degree $1+2p$ with leading coefficient 1.
The first few $S_p$ are 
$S_0=\lambda$, $S_1 = \lambda^3 - \frac{g_2}4 \lambda -\frac{g_3}4$, 
$S_2 = (\lambda^2-3g_2)\bigl(\lambda^3-\frac94 g_2\lambda + \frac{27}4 g_3\bigr)$.

\noindent The following theorem will be deduced from results of~\cite{DN1, Du0, Du2} in Section~\ref{section5}.

\begin{theorem} \label{main2} 
In the case that $C=-\frac{p(p+1)}{2}$~$(p\geq0)$, the product $\sqrt{\frac{S_p(\lambda)}{\lambda}} \, \phi(\lambda, x)$
is a degree~$p$ polynomial in~$\lambda$ with leading coefficient~1 and a degree~$p$ polynomial in $\wp=\wp(x; \tau)$. 
For $p=1$, 
$$
b \= 2 \sqrt{\lambda}\frac{\lambda-\wp}{\sqrt{4\lambda^3 - g_2\lambda - g_3}} \,;
$$
For $p=2$, 
$$
b \= \frac12 \sqrt{\lambda} \frac{4\lambda^2 -12 \wp\,\lambda +36 \wp^2 -9 g_2}{\sqrt{(\lambda^2-3 g_2)(4 \lambda^3 - 9 g_2 \lambda + 27 g_3)}} \, ;
$$
For $p=3$, 
$$
b \= \frac{4 \lambda^3 -24\wp\, \lambda^2 +60(3\wp^2 -g_2)\lambda -900\wp^3 +225 g_2 \wp + 225 g_3}
{\sqrt{16 \lambda^6 -504 g_2 \lambda^4 + 2376 g_3 \lambda^3 +4185 g_2^2 \lambda^2 -36450 g_2 g_3 \lambda -3375 g_2^3 + 91125 g_3^2}} \,. 
$$
\end{theorem}

\subsection*{Organization of the paper} 
In Section~\ref{section2} we review the MR approach.  
In Section~\ref{section3} we prove Theorem~\ref{wavethm}. 
In Section~\ref{section4} we prove Theorems~\ref{BGWnpoint},~\ref{RCC},~\ref{dbessel}, 
and compute explicitly some generalized BGW correlators.
In Section~\ref{section5} we compute Lam\'e correlators.
Further remarks are given in Section~\ref{futherrmk}.

\subsection*{Acknowledgements} 
We would like to express our thanks to Paul Norbury for helpful suggestions 
and to the anonymous referee for many constructive comments.
D.Y. is grateful to Youjin Zhang for his advice and encouragement and to Chang-Shou Lin and Alexander 
Alexandrov for their interest. Part of the work of D.Y. was done in MPIM, Bonn while he was a postdoc; 
he acknowledges MPIM for excellent working conditions and financial supports.

\medskip

\noindent {\it Note added.} 
After this paper was finished we learned that the identities~\eqref{state2}--\eqref{state1}, \eqref{RTheta}--\eqref{1ptC}, 
which had been presented together with their proofs by one of the authors~\cite{Y}, 
had also been found by M.~Bertola and G.~Ruzza (see~\cite{BeR}, whose arXiv version appeared slightly
before the arXiv version of this paper), 
but with a quite different method using matrix models together with the construction of appropriate 
Riemann--Hilbert problems.

\section{Review of the matrix resolvent approach to the KdV hierarchy}\label{section2}
We continue with more details about the MR approach to the study of tau-functions.

\subsection{Fundamental lemma} 
Recall from Introduction some of our notations $\A=\QQ[u_0,u_1,u_2,\dots]$, $\Lambda(\lambda)= \sm 0 1 \lambda 0$, $q=\sm 0 0 {-2u_0} 0$,   
$\cL= \p + \Lambda+q$, and $\cS={\rm sl}_2(\CC) \bllb\lambda^{-1}\brrb$. 
Denote by $(\cdot|\cdot)$ the {\it normalized} Cartan--Killing form. We have 
$(A | B) = {\rm Tr} \, AB$, $\forall\, A,B \in {\rm sl}_2(\CC)$.
This bilinear form naturally extends to $\cS$. 
According to B.~Kostant \cite{Kostant} and V.~Kac \cite{Kac1978}, the space $\cS$ admits the decomposition 
\begin{align}
& \cS \= \Ker \, \ad_\Lambda \oplus \Image \, \ad_\Lambda \,, \quad \Ker \, \ad_\Lambda \perp \Image \, \ad_\Lambda \, , \label{K1}\\
& \Ker \, \ad_\Lambda  \=\biggl \{ \, \sum_{\ell \leq m} c_{\ell} \, \Lambda_{1+2\ell} \, \Big| \, m\in \ZZ \,, \, c_\ell\in\CC \,, c_m\neq 0 \, \biggr\} \,, \label{K2}
\end{align}
where $\Lambda_{1+2\ell}:=\Lambda \, \lambda^\ell$. 
It is easy to see that $\deg \, \Lambda_{1+2\ell}=1+2\ell$ (recall that $\deg$ denotes the principal gradation) and that $(\Lambda|\Lambda)=2\lambda$. 
\begin{lemma}[Drinfeld--Sokolov~\cite{DS}] 
There exists a unique pair $(U,H)$  such that 
\begin{align}
& U\in  \A \, \otimes \, (\Image \, \ad_\Lambda)^{<0}\,, \quad H\in \A \, \otimes \, (\Ker \, \ad_\Lambda)^{<0}\,, \\
& e^{-\ad_U} \cL \=  \p \+ \Lambda \+ H \,. 
\end{align}
\end{lemma}
\noindent The proof, using~\eqref{K1}--\eqref{K2} and the mathematical induction, can be found in~\cite{DS,BDY3}.

\subsection{Matrix resolvent recursive relations} \label{s22} Let us prove that equations in~\eqref{defiR} has a unique solution~$R$.  
The first equation of~\eqref{defiR} reads $ \p (R) + [\Lambda+q \,,\,R]=0$. Write
\beq\label{Rabcd} R \=  \begin{pmatrix} a & b \\  c & -a \\ \end{pmatrix} \, . \eeq 
Then it follows that
\begin{align}
& a \= \frac12 \, \p (b) \, ,\label{aeq}\\
& c \= (\lambda-2u_0) \, b  \,-\,  \frac12 \, \p^2(b) \, ,\label{ceq}\\
& \p^3(b) \,-\, 4  \, (\lambda-2u_0) \, \p(b) \+ 4 \, u_1 \, b \= 0 \,. \label{beq}
\end{align}
Substituting \eqref{aeq}--\eqref{ceq} into the 3rd equation of~\eqref{defiR} we find
\beq\label{essential}
b \, \p^2(b)   \m  \frac1{2} \bigl(\p(b)\bigr)^2 \m  2 \, (\lambda-2u_0)\,b^2 \= -2 \lambda.
\eeq
Applying~$\p$ to both sides of this equation yields \eqref{beq}. This shows compatibility between the 1st and the 3rd equations in~\eqref{defiR}. 
The 2nd equation of~\eqref{defiR} is also compatible with the 3rd one. This gives the existence of~$R$. Let us now prove the uniqueness.  
According to the 2nd equation in~\eqref{defiR} we have
\beq\label{defibk}
b \=  \sum_{k\geq -1}  \frac{b_k }{\lambda^{k+1}} \,, \qquad b_k\in \A \,, \quad b_{-1}\=1 \,.
\eeq
Then \eqref{beq} and \eqref{essential} imply that $b_k$ satisfies the following recursive relations: for $k\geq 0$, 
\begin{align}
& \p \, (b_k) \= 2 \, u_0\, \p \, (b_{k-1}) \+ u_1 \, b_{k-1} \+ \frac14 \p^3 \bigl(b_{k-1}\bigr) \,, \label{mrr1}\\
& b_k \= \sum_{k_1,k_2\geq -1\atop k_1+k_2= k-2}  \biggl[ \biggl(\frac{b_{k_1} \p^2(b_{k_2})}4 - \frac{\p(b_{k_1})  \, \p(b_{k_2})}8 \biggr) +  u_0 \, b_{k_1} b_{k_2}\biggr]  \,-\,  \frac12\sum_{k_1\geq 0, \, k_2\geq -1\atop k_1+k_2= k-2} b_{k_1} b_{k_2+1}  \,, \label{mrr2}
\end{align}
which will be referred as the \textit{matrix resolvent recursive relations}. (One recognizes that~\eqref{mrr1} is the Lenard--Magri recursion.)
Formula~\eqref{mrr2} uniquely determines $b_0,b_1,b_2,\dots$ with the first few being 
\beq\label{firsttwob}
b_0 \= u_0 \,, \quad b_1 \= \frac32 \, u_0^2   \+ \frac1{4} \, u_2 \,.
\eeq
We note that the formal series $b=b(\lambda)$ determined here coincides with~$b$ in the Introduction.
\begin{lemma} 
The basic matrix resolvent can be expressed as $R= e^{{\rm ad}_U}  (\Lambda )$. 
\end{lemma}
\begin{proof} Applying $e^{{\rm ad}_U}$ to 
$\bigl[\p + \Lambda+H \, , \, \Lambda \bigr] \equiv 0 $
gives $\bigl[\cL, e^{{\rm ad}_U}  (\Lambda )\bigr]=0$. Since $\deg \, U<0$, the matrix $R:=e^{{\rm ad}_U} (\Lambda )$ satisfies the 2nd equation in~\eqref{defiR}. 
Finally, ${\rm Tr} \, R^2 = (\Lambda|\Lambda)=2\lambda$. The lemma is proved. 
\end{proof}
 
\begin{defi} \label{dkdefi}
Define a family of derivations~$D_k$~($k\geq 0$) by means of~\eqref{KdVformal} with 
\beq \label{defibQ}Q_k\:=\frac{\p(b_k)}{(2k+1)!!} \,-\, \frac{u_0^k}{k!}u_1 \,. \eeq
\end{defi}

\subsection{From matrix resolvent to tau-structure}
The basic matrix resolvent~$R$ has the expression 
\beq\label{Rbk}
R\= \begin{pmatrix}  \frac12\sum_{i\geq -1}  \p (b_i)  \, \lambda^{-1-i}  & \sum_{i\geq-1} b_i  \, \lambda^{-1-i}  \\ 
\sum_{i\geq -1} \lambda^{-1-i} \Bigl(  (\lambda-2u) \, b_i -\frac{\p^2 (b_i)} 2 \Bigr) &  -\frac12 \sum_{i\geq-1} \p (b_i) \, \lambda^{-1-i}  \end{pmatrix} \,,
\eeq
where we recall that $b_k\in \A$~$(k\geq -1)$ are uniquely determined by~\eqref{mrr2} with~$b_{-1}=1$.
Introduce 
\beq
V_k \:= - \, \frac1{(2k+1)!!} \, \bigl(\lambda^{k} R(\lambda)\bigr)_+ \+ \frac1{(2k+1)!!} \, \begin{pmatrix} 0 & 0 \\ b_k & 0 \end{pmatrix} \,, \qquad k\geq 0\,.
\eeq
\begin{lemma} 
The derivations $D_k$~($k\geq0$) defined in~Definition~\ref{dkdefi} satisfy
\beq\label{commop}
\bigl[ D_k - V_k\,,\, \cL \, \bigr] \= 0 \,, \qquad k\geq 0\,.
\eeq
\end{lemma}
\begin{proof} 
Since $D_k$ commutes with $D_0=\p$, equation~\eqref{commop} can be written equivalently as
$$
D_k \, (q)  \= \bigl[ \, V_k \, , \cL \, \bigr] \,. 
$$
Noting that 
$$V_k 
\= \frac1{(1+2k)!!}
\begin{pmatrix} - \frac12\sum_{i=-1}^{ k -1} \p(b_i) \, \lambda^{k-1-i}  & - \sum_{i=-1}^{k-1} b_i \, \lambda^{k-1-i}  \\ 
\sum_{i=-1}^{ k -1} \lambda^{k-1-i} \Bigl(\frac{ \p^2(b_i)} 2 - (\lambda-2u_0) \, b_i \Bigr) &  \frac12 \sum_{i=-1}^{k-1} \p (b_i) \, \lambda^{k-1-i}  \end{pmatrix}
$$
we have by a direct calculation that 
$$
\bigl[V_k \, , \, \cL \bigr] \= \bigl[V_k \, ,  \, \p  + \Lambda + q \bigr] 
\=  \frac1{(1+2k)!!} \begin{pmatrix} 0 & 0 \\   -2 \, \p(b_k) &   0 \end{pmatrix} \,.
$$
Here \eqref{mrr1} has been used for the simplification. 
The lemma is proved. 
\end{proof}

\begin{defi} 
Define $P_k= e^{\ad_U} (\Lambda_{1+2k})$, $k=0,1,2,3,\dots$.
\end{defi}

\begin{lemma} The following formula holds true
\beq\label{dkp}
 D_k (P_\ell) \= \bigl[ V_k \, , P_{\ell} \bigr] \,, \quad \forall \, k,\ell \geq 0 \,.
\eeq
\end{lemma}
\begin{proof}  
Applying $e^{-\ad_U}$ to the both sides of~\eqref{commop} yields
$$
\bigl[D_k  + S_k \,, \, \p+ \Lambda+ H \bigr] \= 0 \,,
$$
where $S_k:= \sum_{i\geq 0} \frac{(-1)^i}{(i+1)!} \ad_U^i \bigl(D_k(U)\bigr) - e^{-\ad_U} (V_k)$. 
Clearly, $S_k$ belongs to $\A_u\otimes \cS$. Decompose $S_k=S_k^{\Ker}+S_k^{\Image}$ with 
$S_k^{\Ker}\in \A\otimes\Ker_{{\rm ad}_\Lambda}$,~$S_k^{\Image}\in \A\otimes\Image_{{\rm ad}_\Lambda}$. 
We have 
\begin{align}
& \p \bigl(S_k^{\Ker}\bigr) \=  D_k \bigl(H\bigr)  \,, \label{s1k}\\
& \p \bigl(S_k^{\Image}\bigr) \= \Bigl[S_k^{\Image} \,,\,\Lambda+H\Bigr] \,.\label{s2i}
\end{align}
Equation \eqref{s2i} implies that $S_k^{\Image}$ must vanish. Otherwise, write $S_k^{\Image}=\sum_{j\leq m} S_k^{\Image,[j]}$ for some $m$ with $S_k^{\Image,[m]}$ is not 
zero. Then since $\deg H<0$, the highest (principal) degree equation of \eqref{s2i} reads $[S_k^{\Image,[m]}, \Lambda]=0$. This implies that $S_k^{\Image,[m]}=0$, which produces 
a contradiction. Therefore $S_k$ belongs to $\A\otimes\Ker \, {\rm ad}_\Lambda$. This implies the idenity 
$$
\bigl[ \,D_k + S_k\, , \, \Lambda_{1+2\ell} \, \bigr] \=0\,.
$$
Applying $e^{\ad_{U}}$ to this identity yields~\eqref{dkp}.
\end{proof}

Let us proceed to the tau-structure. We first verify that $\Omega_{p,q}$ are well-defined from eq.~\eqref{tauom}. Indeed, 
$$
R(\mu) \= R(\lambda) \+ R'(\lambda) (\mu-\lambda) \+ (\mu-\lambda)^2 \, \p_\lambda \biggl(\frac{R(\lambda)-R(\mu)}{\lambda-\mu}\biggr) \,,
$$
so we have 
$$
\mbox{RHS of}~\eqref{tauom} \= \frac{2\lambda}{(\lambda-\mu)^2} \,-\, \frac{\tr \bigl (R(\lambda) R'(\lambda)\bigr)}{\lambda-\mu} 
\+ \tr \, \biggl( R(\lambda) \, \p_\lambda \Bigl(\frac{R(\lambda)-R(\mu)}{\lambda-\mu}\Bigr)\biggr) \,-\, \frac{\lambda+\mu}{(\lambda-\mu)^2} \,.
$$
It is obvious from the property 
${\rm tr} \, R(\lambda)^2 = 2 \lambda$ that $\tr \bigl (R(\lambda) R'(\lambda)\bigr)=1$. 
Therefore $\mbox{RHS of}~\eqref{tauom}=\tr \, \Bigl( R(\lambda) \, \p_\lambda 
\Bigl(\frac{R(\lambda)-R(\mu)}{\lambda-\mu}\Bigr)\Bigr)\in
\mathcal{A} \bllm \lambda^{-1},\mu^{-1} \brrm \, \lambda^{-1}\mu^{-1} $.
This finishes the verification.

Introduce $\nabla(\lambda) \,= \, - \,  \sum_{k\geq 0} \frac{(1+2k)!!}{\lambda^{1+k}} D_k$.
\begin{lemma}\label{lemmanablaR} The following identity holds true:
\beq
\nabla(\lambda) \bigl (R(\mu)\bigr) \= \frac{[R(\lambda),R(\mu)]}{\lambda-\mu}  \,-\, \bigl[Q(\lambda),R(\mu)\bigr]\,, \quad 
Q(\lambda):= \begin{pmatrix} 0 & 0 \\ b(\lambda) & 0 \end{pmatrix} \,. 
\eeq
\end{lemma}

\begin{proof} Using equation~\eqref{dkp} with $\ell=0$ we have
\begin{align}
\nabla(\lambda) \bigl (R(\mu)\bigr) 
& \= - \, \sum_{k\geq 0} \frac{1}{\lambda^{1+k}} \biggl[ \, - \bigl(\mu^k R(\mu)\bigr)_+ + \begin{pmatrix} 0 & 0 \\ b_k & 0 \end{pmatrix} \,,\, R(\mu) \, \biggr] \nn\\
& \=  \frac1{2\pi i} \oint_{|\mu|<|\rho|<|\lambda|} d\rho \frac{\bigl[R(\rho), R(\mu)\bigr]}{(\lambda-\rho)(\rho-\mu)} 
\,-\, \biggl[ \, \begin{pmatrix} 0 & 0 \\ b(\lambda)-1 & 0 \end{pmatrix} \,,\, R(\mu) \, \biggr]  \nn\\
& \= \frac{[R(\lambda),R(\mu)]}{\lambda-\mu}   \,-\, \Bigl[ \, {\rm Coef}_{\lambda^1} R(\lambda) \,, \, R(\mu) \, \Bigr] 
\,-\, \biggl[ \, \begin{pmatrix} 0 & 0 \\ b(\lambda)-1 & 0 \end{pmatrix} \,,\, R(\mu) \, \biggr] \,. \nn
\end{align} The lemma is proved.
\end{proof}

\noindent {\it Proof} of Proposition \ref{taulemma}.  The fact that $\Omega_{0,0}=u$ can be obtained from a straightforward residue computation. 
Since the RHS of \eqref{tauom} is invariant with respect to $\lambda\leftrightarrow\mu$, we have $\Omega_{p,q}=\Omega_{q,p}$.  
Applying $\nabla(\nu)$ to the both sides of \eqref{tauom} we obtain
\begin{align}
\sum_{p,q,r\geq 0} D_r\bigl(\Omega_{p,q}\bigr) \frac{(1+2p)!!(1+2q)!!(1+2r)!!}{\lambda^{p+1} \mu^{q+1} \nu^{r+1} } \=  &
\frac{{\rm Tr} \,[R(\nu),R(\lambda)] R (\mu)}{(\lambda-\mu)^2(\nu-\lambda)} \+ \frac{{\rm Tr} \, R(\lambda) \, [R(\nu), R (\mu)]}{(\lambda-\mu)^2(\nu-\mu)}  \nn\\
&   \, - \, \frac{{\rm Tr} \,[Q(\nu),R(\lambda)] R (\mu)}{(\lambda-\mu)^2} \,-\, \frac{{\rm Tr} \, R(\lambda) \, [Q(\nu), R (\mu)]}{(\lambda-\mu)^2} \nn\\
 \= &  - \frac{{\rm Tr} \,[R(\nu),R(\lambda)] R (\mu) }{(\lambda-\mu)(\mu-\nu)(\nu-\lambda)} \,.\nn
\end{align}
The RHS of this identity is invariant with respect to the permutations of $\lambda,\mu,\nu$, so is the LHS. This implies that 
$D_r\bigl(\Omega_{p,q}\bigr) = D_p\bigl(\Omega_{q,r}\bigr) = D_q\bigl(\Omega_{r,p}\bigr)$. 
For $n=3$, the statement is already proven in the above proof of Proposition~\ref{taulemma}. 
Formula~\eqref{thmformula} can be proved by induction; details can be found in \cite{BDY1,BDY3} and are omitted here.  The proposition is proved. 
\epf

\begin{cor}\label{integrableascor}
The $D_k$~$(k\geq 0)$ (defined in Definition~\ref{dkdefi}) commute pairwise. Moreover, we have $Q_0=0$, $Q_1=u_3/12$, and $Q_k\in \A^{\geq 2}$. 
\end{cor}
\begin{proof}
We have for all $k,\ell\geq 0$, 
$$
D_k D_\ell \bigl(\Omega_{0;0}\bigr) \= D_k D_0 \bigl(\Omega_{\ell;0}\bigr)
\= D_0 D_k \bigl(\Omega_{\ell;0}\bigr) \= D_0 D_0  \bigl(\Omega_{k;\ell}\bigr)  \= D_\ell D_k \bigl(\Omega_{0;0}\bigr)\,.
$$
Therefore $[D_k,D_\ell]=0$. 
The fact that $Q_0=0$ and $Q_1=u_3/12$ can be easily seen from \eqref{firsttwob} and~\eqref{defibQ}. 
Define $B_k=b_k|_{u_i=0, \, i\geq 1}$. Then $B_{-1}=1$. 
Equation~\eqref{mrr2} implies that 
$$
B_k \= \sum_{k_1,k_2\geq -1\atop k_1+k_2= k-2}    u_0 \, B_{k_1} B_{k_2} \,-\,  \frac12\sum_{k_1\geq 0, \, k_2\geq -1\atop k_1+k_2= k-2} B_{k_1} B_{k_2+1}  \,, \quad k\geq 0 \,.
$$
It follows from this recursion that $B_k= \frac{(1+2k)!!}{(k+1)!} u_0^{k+1}$ for all $k\geq 0$, which implies $Q_k\in \A^{\geq 2}$. 
The corollary is proved. 
\end{proof}

Corollary~\ref{integrableascor} and the uniqueness statement proved in~\cite{LZ} imply that 
the $D_k$ defined in Definition~\ref{dkdefi} coincide with the ones defined in Section~\ref{section11}. 
Our definition of tau-function (see Section~\ref{MRI}) of any solution also 
coincides with the one in the literature \cite{DJKM,DKJM,DZ-norm}, which can be seen by observing that $$\Omega_{p,q}^{[0]} \= \frac{u^{p+q+1}}{p! \, q! \, (p+q+1)} \,, $$
where $\Omega_{p,q}^{[0]}$ denotes the degree zero part of~$\Omega_{p,q}$.

\section{From a pair of wave functions to the $n$-point functions}\label{section3}
Denote by $L=\p^2 + 2u_0:\A \rightarrow \A$ the Lax operator for the KdV hierarchy. Let us first briefly recall the 
pseudo-differential operator approach~\cite{Dickey} to the study of the KdV equation.  Notations will be the same as in Section~\ref{Introwave}. Denote 
$A_k = \frac{1}{(2k+1)!!} \bigl(L^{\frac{2k+1}2}\bigr)_+$~($k\geq 0$). 
We recall that $A_k$ is a differential operator with coefficients in~$\A$.
Define 
a family of derivations $D_k$ by
$D_k (u_0) := \bigl[A_k, L / 2\bigr]$.
It is straightforward to show (by using the properties of pseudo-differential operators) 
that these $D_k$ commute pairwise with $D_0=\p$, $D_1(u_0)= u_0 u_1 + u_3/12$ and $D_k(u_0)=u_0^k u_1/k! + \cdots$, where the dots denotes terms in $\A^{\geq 2}$. 
So these $D_k$ coincide with the 
ones defined in~Section~\ref{section11}, 
and the KdV hierarchy can be written equivalently as 
\beq
\frac{\p L}{\p t_k} \=   
\bigl[A_k, L \bigr] \,  \label{KdVLA}
\eeq
with $\p$ replaced by $\p_x$ and $u_i$ by $\p_x^i(u)$. 
In the rest of this section, we give a way of constructing a pair of wave functions of an arbitrary solution, 
and apply it as well as Corollary~\ref{cor2} to prove Theorem~\ref{wavethm}. 

\subsection{Dressing operators and pairs of wave functions}  
Similarly as in Section~\ref{Introwave}, we start with the {\it time-independent} case.  Namely, consider $L=\p^2+2f(x)$ 
with arbitrarily given fixed $f(x)\in V$. 
Here we recall that $V$ is a ring of functions of~$x$ closed under $\p_x$.
A pseudo-differential operator $\Phi$ of the form $\Phi=\sum_{m\geq 0} \phi_m(x) \, \p_x^{-m}$ with $\phi_0(x)\equiv 1$
is called a {\it dressing operator} of~$L$ if
$$
\Phi(x) \circ \p_x^2 \circ \Phi(x)^{-1} \= L \,. 
$$ 
Here $\phi_m(x)$~($m\geq1$) will live in a certain ring $\wt V$ satisfying $V\subseteq \p_x \bigl(\wt V\bigr)\subseteq \wt V$. 

If $\Phi(x)$ exists then the  
freedom of~$\Phi(x)$ is given by an operator of the form 
$\sum_{m\geq 0} g_m \, \p_x^{-m}$ for arbitrary 
constants $g_m$ with $g_0=1$ through the right composition 
\beq\label{freedomg}
\Phi(x)  \quad \mapsto \quad \Phi(x) \circ \sum_{m\geq 0} g_m \, \p_x^{-m} \,. 
\eeq
Recall that for pseudo-differential operators~\cite{Dickey}, most operations (e.g., composition) 
only used differentiations in~$x$, namely, for $a\in \CC\smallsetminus\NN$ 
we usually think of~$\p_x^a$ as a formal symbol and 
do not actually need to consider how it acts on a function of~$x$. 
But, following~\cite{Dickey}, for all $p(x)\in\wt V$ we can define $\p_x^{-m} \bigl(p(x) \, e^{xz}\bigr)$  
for any $m\in \CC$ as $\sum_{l\geq 0} \binom{-m}{l} \p_x^l(p(x)) \, z^{-m-l}\, e^{xz}$.
If $\Phi=\Phi(x)=\sum_{m\geq 0} \phi_m(x) \, \p_x^{-m}$ is 
a dressing operator of~$L$, 
then it is straightforward to check that 
$\psi := \Phi(x) \bigl(e^{xz}\bigr) $ is in $\wt V \bllb z^{-1}\brrb\, e^{xz}$ and is a wave function of~$f$ (cf. Section~\ref{Introwave}). 
Explicitly, $\psi =  \sum_{m\geq 0} \phi_m(x) \, z^{-m} \, e^{xz} $. 
On the contrary, if $\psi =  \sum_{m\geq 0} \phi_m(x) \, z^{-m} \, e^{xz} $ is a wave function of~$f$, 
then $\Phi:=\sum_{m\geq 0} \phi_m(x) \, \p_x^{-m}$
is a dressing operator of~$L$.  This gives the 1-1 correspondence between the dressing operators of~$L$ and the wave functions of~$f$.  
Also, the freedom of~$\Phi$ generated by the right composition~\eqref{freedomg} matches to 
$\psi\mapsto g(z) \, \psi$ with $g(z)=\sum_{m\geq0} g_m \, z^{-m}$. 

Let us give a proof of the existence of a wave function. Write
\beq\label{psibeta}
\psi(z,x) \= e^{ \beta(z,x) } \, e^{xz} \,, \qquad \beta(z,x) \:= \sum_{ j \geq 1} \beta_j(x) \, z^{-j}  \,. 
\eeq
The defining equation $L(\psi)=z^2 \psi$ for the wave function $\psi$ written in terms of~$\beta$ is given by
\beq
\beta_{xx} \+ \beta_x^2   \+ 2  \, z \,   \beta_x    \+ 2f  \= 0 \,. 
\eeq
Substituting the expression~\eqref{psibeta} of~$\beta$ in this equation we obtain
\beq\label{betarec}
 \beta_{j}'(x) \= - \frac12 \, \beta_{j-1}''(x)  \,-\, \frac12 \sum_{ j_1,j_2 \geq 1 \atop j_1+j_2 = j-1} \beta_{j_1}'(x)  \beta_{j_2}'(x)  \,-\, f  \, \delta_{j,1} \,,
\eeq
where $\beta_0(x)$ is defined to as the $0$-function and ${}'=d/dx$. Clearly, the recursion~\eqref{betarec}
uniquely determines $\beta_j'(x)$ for all $j\geq 1$ and $\beta_j'(x)\in V$. 
Since $V\subseteq \p_x \bigl(\wt V\bigr)$, we conclude that the solution $\psi$ of the form~\eqref{psibeta} to the equation $L(\psi)=z^2 \psi$ exists with  
$\beta_j(x) \in \wt V$. Expanding the exponential $e^{\beta(z,x)}$ into Taylor series of~$z^{-1}$ we obtain finally the required form 
$\psi =  \sum_{m\geq 0} \phi_m(x) \, z^{-m} \, e^{xz}$ with $\phi_0(x)\equiv 1$ and $\phi_m(x)\in \wt V$ for a wave function. 

Now let $\psi$ be a wave function of~$f$ and $\Phi$ the (unique) dressing operator of~$L$ associated to~$\psi$.  
Define a particular element $\psi^*$ in $\wt V \bllb z^{-1}\brrb\, e^{-xz}$ by
\beq\label{defpsistar}
\psi^* \= \psi^*(z,x) \:= \bigl(\Phi^{-1}\bigr)^* \, \bigl(e^{-xz}\bigr) \,, 
\eeq
where $\bigl(\Phi^{-1}\bigr)^*$ denotes the formal adjoint operator of~$\Phi^{-1}$.
For any $y_1 = a_1 \, e^{xz} \in \wt V \bllb z^{-1}\brrb\, e^{xz}$ and $\tilde y_1 = b_1 \, e^{xz} \in \wt V \bllb z^{-1}\brrb\, e^{-xz}$, define  
$y_1 \tilde y_1 := a_1 b_1 \in \wt V \bllb z^{-1}\brrb $. 

\begin{lemma} 
The elements $\psi$ and $\psi^*$ form a pair of wave functions of~$f$. 
\end{lemma}
\begin{proof}
It is easy to check that $\psi^*$ is indeed of the form
$\psi^* = \bigl(1+ \phi_1^*(x)/z+\phi_2^*(x)/z^2 + \cdots\bigr) \,e^{-xz}$ with $\phi^*(x)\in \wt V$
and satisfies $L (\psi^*) = z^2 \psi^*$.  We are left to show that for all $i\geq 0$, 
$$
\res_{z=\infty} \, \p_x^i\bigl(\psi(z,x)\bigr) \, \psi^*(z,x) \, dz \= 0\,. 
$$
Note that for any two pseudo-differential operators $P,Q$ we have
\beq
\res_{z=\infty} \, P \bigl(e^{xz}\bigr) \, Q \bigl(e^{-xz}\bigr) \, dz  \= \res_{\p_x} \,  P \circ Q^* \,.
\eeq
Taking $P=\p_x^i \circ \Phi$ and $Q=\Phi^{-1}$ in this identity we find that
$$
 \res_{z=\infty} \, \p_x^i\bigl(\psi(z,x)\bigr) \, \psi^*(z,x) \, dz  \= \res_{z=\infty} \, \p_x^i \circ \Phi \bigl(e^{xz}\bigr) \, \bigl(\Phi^{-1}\bigr)^* (e^{-xz}) \, dz \=  \res_{\p_x} \,  \p_x^i \circ \Phi \circ \Phi^{-1} \= 0   \, .
$$
The lemma is proved. 
\end{proof}

Hence we have proved the existence of {\it a pair of wave functions} of~$f$. 

Observe that the above proof is revertible. Namely, the zero-residue condition ensures the uniqueness of the 
dual wave function associated with~$\psi$. It must be the $\psi^*$ defined by~\eqref{defpsistar}.

\begin{lemma}\label{twoid}
We have $\psi_x \psi^* - \psi \psi^*_x = 2z $ and $\psi(z,x) \, \psi^*(z,x) = b(z^2,x)$, with~$b$ as in~\eqref{Rabcd}.
\end{lemma}
\begin{proof}
Denote $W=\psi_x \psi^* - \psi \psi^*_x$. We have $W_x =0$. It follows that $W$ must have the form
$W=2z + \sum_{k\geq 0} s_k z^{-k}$, where $s_k$ are constants. For any $i\geq -1$, 
\begin{align}
\res_{z=\infty } \, z^i \, W \,  dz 
& \=  \res_{z=\infty} \Bigl ( \p_x \circ \Phi \circ \p_x^i  (e^{xz})  \, \bigl(\Phi^{-1}\bigr)^* (e^{-xz})  -  (-1)^i  \, \Phi (e^{xz})  \,  \p_x \circ \bigl(\Phi^{-1}\bigr)^* \circ \p_x^i (e^{-xz}) \Bigr) \, dz \nn\\
& \=  \res_{\p_x} \Bigl ( \p_x \circ \Phi \circ \p_x^i  \circ \Phi^{-1}  +   \Phi  \circ \p_x^i  \circ \Phi^{-1} \circ \p_x \Bigr)  \= 
\res_{\p_x} \Bigl ( \p_x \circ L^{\frac{i}2}  +   L^{\frac{i}2} \circ \p_x \Bigr) \= 0\,. \nn
\end{align}
The last equality uses the fact that $L$ is a differential operator and that $(L^*)^{1/2}=-L^{1/2}$. Therefore all $s_k$ vanish. The first identity is proved. 
To prove the second equality we observe that both the LHS and the RHS satisfies the same differential equation~\eqref{essential} (with $u_0$ replaced by $f$ and $\p$ by $\p_x$).
Moreover, they are both formal power series in $z^{-1}$ with the leading term~1. Then the uniqueness statement implies the second identity. The lemma is proved. 
\end{proof}

We note that the identity $\psi(z,x) \, \psi^*(z,x) = b(z^2,x)$ can be used as an 
alternative and possibly useful criterion for a pair of wave functions. 
We should also notice that for an arbitrary solution $u(\bt)$ there is a particular subclass of pairs of wave functions 
of~$u(\bt)$  
with the additional property $\psi(-z,x)=\psi^*(z,x)$, which was often assumed in the literature.

Let us now proceed to the {\it time-dependent} case, i.e. to prove the existence of 
 a pair of wave functions of an arbitrary solution of the KdV hierarchy in the ring $V\llm \bt_{>0}\rrm $.
We recall that for the operators $L=\p^2 + 2u_0$ and $A_k$, once a solution $u=u(\bt)$ is taken, we will 
replace $\p$ by $\p_x$ and $u_i$ by $\p_x^i(u)$. 
\begin{prop} \label{dressingprop}
Let $u=u(\bt)$ be a solution to the KdV hierarchy~\eqref{KdVhk} in the ring 
$V\llm \bt_{>0}\rrm $. There exists a pseudo-differential operator~$\Phi$ of the form 
\beq \label{phimprop}
\Phi \= \Phi(\bt) \= \sum_{m\geq 0}  \phi_m(\bt) \, \p_x^{-m}\,, \qquad \phi_0(\bt)  \, \equiv \, 1
\eeq
such that   
\begin{align}
& L \= \Phi \circ \p_x^2 \circ \Phi^{-1} \,,   \label{dressing1}  \\
&  \p_{t_k} \bigl(\Phi\bigr)  \=  -  \frac1{(2k+1)!!} \bigl(\Phi \circ \p_x^{2k+1} \circ \Phi^{-1} \bigr)_-  \circ \Phi \,, 
 \qquad \forall\, k\geq 0\,.   \label{dressing2}
\end{align}
Here, $t_0=x$ as usual and $\phi_m(\bt)\in \wt V \llm \bt_{>0}\rrm $~$(m\geq 1)$. 
\end{prop}
\begin{proof}
Let $f(x):=u(x,\bdzero)$ be the initial data of the solution. According to the above time-independent theory, we can take a
 pseudo-differential operator $\Phi(x)$ of the form $\Phi(x) = \sum_{m=0}^\infty \phi_m(x) \, \p_x^{-m}$ 
 with $\phi_0(x)\equiv 1$ and $\phi_m(x)\in \wt V$~$(m\geq 1)$, such that 
$ \p_x^2 \+  2f(x) \= \Phi(x) \circ \p_x^2 \circ \Phi(x)^{-1} $.

Now consider the initial value problem for~$\Phi(\bt)$ given by equations~\eqref{dressing1} and~\eqref{dressing2} with the initial data 
$\Phi(x,\bdzero) = \Phi(x)$. 
First let us check the compatibility between \eqref{dressing1} and the $k=0$ case of equations~\eqref{dressing2}.  
We notice that $(\Phi \circ \p_x \circ \Phi^{-1})_-=\bigl(L^{1/2}\bigr)_- = L^{1/2}-\bigl(L^{1/2}\bigr)_+=
\Phi \circ \p_x \circ \Phi^{-1} - \p_x$. Therefore, 
$$\p_{t_0} (\Phi)  \=  - \, \Phi \circ \p_x \+ \p_x \circ \Phi \=  \p_x (\Phi)\,.$$ 
This is compatible with our convention $t_0=x$. Secondly, we will prove the compatibility 
between~\eqref{dressing1} and equations~\eqref{dressing2} with $k\ge1$.  Indeed, 
we have 
\begin{align}
 \frac{\p L}{\p t_k} & \=  \p_{t_k}(\Phi) \circ \p_x^2 \circ \Phi^{-1}  \,-\, \Phi \circ \p_x^2 \circ \Phi^{-1} \circ \p_{t_k}(\Phi)  \circ \Phi^{-1} \nn\\
 &  \= -  \frac1{(2k+1)!!} \, \Bigl(L^{\frac{2k+1}2} \Bigr)_-  \circ L  \+   \frac1{(2k+1)!!} \, L \circ \Bigl(L^{\frac{2k+1}2} \Bigr)_- \=   
 \bigl[A_k, L \bigr] \,.\nn
\end{align}
This is true as $u$ is a solution of the KdV hierarchy (cf.~\eqref{KdVLA}). 
Finally, we check the compatibility between all the equations of~\eqref{dressing2}. Indeed, 
\begin{align}
& \p_{t_\ell} \p_{t_k} (\Phi)  \,-\, \p_{t_k} \p_{t_\ell} (\Phi)  \nn\\
& \quad \=    
 \frac{ \Bigl( \bigl[\bigl(L^{\frac{2\ell+1}2}\bigr)_- , L^{\frac{2k+1}2}\bigr]  \,-\, \bigl[\bigl(L^{\frac{2k+1}2}\bigr)_- , L^{\frac{2\ell+1}2}\bigr]  \+ \bigl[\bigl(L^{\frac{2k+1}2} \bigr)_-  ,  \bigl(L^{\frac{2\ell+1}2}\bigr)_- \bigr]  \Bigr)_-  }{(2k+1)!! \, (2\ell+1)!!} \circ \Phi    \nn\\ 
& \quad \=   \frac{ \Bigl( \bigl[\bigl(L^{\frac{2\ell+1}2}\bigr)_- , L^{\frac{2k+1}2}\bigr]  \,-\, \bigl[L^{\frac{2k+1}2} , \bigl(L^{\frac{2\ell+1}2} \bigr)_+\bigr]   \Bigr)_-  }{(2k+1)!! \, (2\ell+1)!!} \circ \Phi 
 \=   \frac{ \Bigl( \bigl[L^{\frac{2\ell+1}2} , L^{\frac{2k+1}2}\bigr]   \Bigr)_-  }{(2k+1)!! \, (2\ell+1)!!} \circ \Phi  \=  0 \,. \nn 
\end{align}
The proposition is proved. 
\end{proof}

We call $\Phi(\bt)$ in Proposition~\ref{dressingprop} a {\it dressing operator} of~$L=\p_x^2+2u(\bt)$.
It is clear from the proof of Proposition~\ref{dressingprop} that
the freedom of the dressing operator~$\Phi(\bt)$ is characterized by that of~$\Phi(x)$, i.e., by
a sequence of arbitrary constants $g_1,g_2,\dots$ through  
$$
\Phi \quad \mapsto \quad \Phi \circ  \sum_{m = 0}^\infty g_m \, \p_x^{-m} \,, \quad g_0 \= 1\,. 
$$
Similarly as in the time-independent case we will use the dressing operator to 
prove the existence of a pair of wave functions of~$u$. As in~\eqref{phimprop}, write 
$\Phi(\bt) = \sum_{m\geq 0}  \phi_m(\bt) \, \p_x^{-m}$ with $\phi_0(x) = 1$.
Put  $Q = e^{\sum_{k\geq 0} \frac{z^{2k+1}}{(2k+1)!!} t_k}$. 
Define for any $p(\bt)\in \wt V\llm \bt_{>0} \rrm$ and for all $m\in \ZZ$, 
$\p_x^{-m} \Bigl(p(\bt) \, e^{\sum_{k\geq 0} \frac{z^{2k+1}}{(2k+1)!!} t_k}\Bigr) :=  \sum_{l\geq 0} \binom{-m}{l} \p_x^l(p(\bt)) \, z^{-m-l}\, e^{\sum_{k\geq 0} \frac{z^{2k+1}}{(2k+1)!!} t_k}$.
Then one can check that the following element 
\beq
\psi(z,\bt) \:=\Phi(\bt) \, (Q) \= \biggl(1 \+ \frac{\phi_1(\bt)}{z} \+ \frac{\phi_2(\bt)}{z^2} \+  \cdots \biggr) \, e^{\sum_{k\geq 0} \frac{z^{2k+1}}{(2k+1)!!} t_k}
\eeq
is a wave function of~$u$. For any $p(\bt)\in \wt V\llm \bt_{>0} \rrm$ and $m\in \ZZ$, define
$\p_x^{-m} \Bigl(p(\bt) \, e^{ - \sum_{k\geq 0} \frac{z^{2k+1}}{(2k+1)!!} t_k}\Bigr) :=  \sum_{l\geq 0} \binom{-m}{l} \p_x^l(p(\bt)) \, (-z)^{-m-l}\, e^{-\sum_{k\geq 0} \frac{z^{2k+1}}{(2k+1)!!} t_k}$.
For any $y_1 = a_1 \, e^{\sum_{k\geq 0} \frac{z^{2k+1}}{(2k+1)!!} t_k}$, 
$\tilde y_1 = \tilde{a}_1 \,e^{-\sum_{k\geq 0} \frac{z^{2k+1}}{(2k+1)!!} t_k}$ 
 with $a_1$ and~$\tilde{a}_1$ in $\wt V\llm \bt_{>0}\rrm  \bllb z^{-1}\brrb$, we define  
$y_1 \tilde y_1 := a_1 b_1 \in \wt V\llm\bt_{>0}\rrm  \bllb z^{-1}\brrb$. 
One can then verify that the following element 
\beq
\psi^*(z,\bt)  \:= \bigl(\Phi(\bt)^{-1}\bigr)^* \, \bigl(Q^{-1}\bigr) \= \biggl(1 \+ \frac{\phi_1^*(\bt)}{z} \+ \frac{\phi_2^*(\bt)}{z^2} \+  \cdots \biggr) \, e^{-\sum_{k\geq 0} \frac{z^{2k+1}}{(2k+1)!!} t_k} 
\eeq
is the dual wave function of~$u$ associated with~$\psi$. Hence we have proved the existence of {\it a pair of wave functions} of~$u$.

\subsection{Proof of Theorem~\ref{wavethm}}\label{section3point2} 
Let $u=u(\bt)$ be an arbitrary solution to the KdV hierarchy~\eqref{KdVhk} and $(\psi,\psi^*)$
a pair of wave functions of~$u$. We first prove a useful lemma. 
\begin{lemma} \label{factorize1} Define 
\beq
\Psi (z,\bt) \= \begin{pmatrix} \psi (z,\bt) & \psi^*(z,\bt) \\ -\psi_x(z,\bt) & -\psi^*_x(z,\bt) \end{pmatrix}\,. 
\eeq
Then we have 
\beq \label{wronskian}
 \det \Psi(z,\bt) \; \equiv\;  2z 
\eeq
and 
\beq\label{bpsipsistaru}
b (\lambda, \bt) \= \psi(z, \bt) \, \psi^*(z,\bt) \,,
\eeq
where $\lambda=z^2$.
Moreover, the matrix resolvent~$R(\lambda,\bt)$ of~$u$ is given by
\beq\label{RP}
R(\lambda,\bt) \; \equiv \;  - \, \Psi(z,\bt) \, \begin{pmatrix} z & 0 \\ 0 & -z \end{pmatrix} \, \Psi^{-1}(z,\bt)  \,.  
\eeq
\end{lemma}
\begin{proof}
The proof for~\eqref{wronskian} and~\eqref{bpsipsistaru} is almost identical with that for Lemma~\ref{twoid}, so we omit the details. 
Let us now prove \eqref{RP}.  
We have 
\beq
 - \, \Psi(z,\bt) \, \begin{pmatrix} z & 0 \\ 0 & -z \end{pmatrix} \, \Psi^{-1}(z,\bt) \= 
 \begin{pmatrix}  \frac12 (\psi \psi^*)_x & \psi \, \psi^* \\  -\psi_x \psi_x^* & -\frac12 (\psi\psi^*)_x \end{pmatrix}  \,. 
\eeq
So \eqref{bpsipsistaru} shows the $(1,2)$-entry identity of~\eqref{RP}. 
The $(1,1)$-entry identity and $(2,2)$-entry identity are also true due to~\eqref{aeq}. 
It remains to show 
\[ - \psi_x \psi_x^* \= (\lambda-2u) \, b  \,-\,  \frac12 \, \p_x^2(b) \,. \]
Indeed, 
\[
{\rm RHS} \= (\lambda-2u) \, \psi  \psi^*  \,-\,  \frac12 \, (\psi_{xx} \psi + 2\psi_x \psi^*_x + \psi \, \psi^*_{xx} ) \= {\rm LHS} \,. 
\]
The lemma is proved. 
\end{proof}

According to equation~\eqref{RP}  we can write the basic matrix resolvent $R(z^2,\bt)$ of~$u$ in terms of $\psi,\psi^*$ as 
\begin{align}
& R\bigl(z^2,\bt\bigr)
\= z \+ \begin{pmatrix}  \psi \\ -\,\psi_x \end{pmatrix} \begin{pmatrix}   \psi_x^* & \psi^* \end{pmatrix}  \,.  \label{Rpsiz}
\end{align}
For simplicity, we will often denote $R(\lambda)=R(\lambda,\bt)$. Through a direct calculation we have 
\begin{align}
& \tr \, R\bigl(z_1^2\bigr) \, R\bigl(z_2^2\bigr) \=
 \bigl(\psi_x^*(z_1) \, \psi(z_2)-\psi^*(z_1) \psi_x(z_2)\bigr)
\bigl(\psi_x^*(z_2) \, \psi(z_1)-\psi^*(z_2) \psi_x(z_1)\bigr) \, - \, 2z_1z_2   \,.
\end{align}
Hence 
\begin{align}
& \sum_{p_1,p_2} \Omega_{p_1, p_2} (\bt) \,  \frac{(2p_1+1)!! \, (2p_2+1)!!}{z_1^{2p_1+2} z_2^{2p_2+2}}  \= 
- \, D(z_1,z_2;\bt)  \, D(z_2,z_1;\bt) \,- \, \frac{1}{(z_1-z_2)^2} \,. \label{newformula2point}
\end{align}
This shows formula~\eqref{newformula} for $n=2$. For $n\geq 3$, we are going to use a formula given in~\cite{DuY1}:
\begin{align} 
&\sum_{\sigma\in S_n/C_n} 
\frac{ {\rm tr} \, R(\lambda_{\sigma(1)})\dots R(\lambda_{\sigma(n)}) } {\prod_{i=1}^n \bigl(\lambda_{\sigma(i+1)}-\lambda_{\sigma(i)}\bigr)}  \nn\\
& \quad \=\sum_{\sigma\in S_{n-2}} \frac{\bigl\langle R(\lambda_n) \,,\, {\rm ad}_{R(\lambda_{\sigma(1)})} \cdots {\rm ad}_{R(\lambda_{\sigma(n-2)})} 
R(\lambda_{n-1})  \bigr\rangle}{ \bigl(\lambda_{n-1}-\lambda_{\sigma(n-2)}\bigr) (\lambda_{n}-\lambda_{n-1}) \bigl(\lambda_{\sigma(1)}- \lambda_n\bigr)\prod_{i=1}^{n-3} \bigl(\lambda_{\sigma(i+1)}-\lambda_{\sigma(i)}\bigr)} \,, \nn
\end{align}
where ${\rm ad}_a \, b := [a,b]$ and $\langle a,b\rangle := \tr \, ab$. Here we have abbreviated $R_i(\lambda)=R_i(\lambda,\bt)$.
This formula tells that the contributions to the $n$-point generating series coming from the term~$z$ in~\eqref{Rpsiz} are zero. 
Therefore the theorem is proved by using the same argument as in~\cite{DYZ}.  \epf

\subsection{Proof of Theorem~\ref{thmmainabs}} 
We will first prove Proposition~\ref{KDrelation}.
It follows easily from Lemma~\ref{factorize1} that the fours functions $\psi,\psi^*,\psi_x,\psi^*_x$ satisfy the following three relations:
\begin{align}
& \psi(z,\bt) \, \psi^*(z,\bt) \= b(z^2,\bt) \,, \nn\\
& \psi_x(z,\bt) \, \psi^*(z,\bt) \,-\, \psi(z,\bt) \, \psi^*_x(z,\bt) \= 2z \,, \nn\\
& \psi_x(z,\bt) \, \psi^*(z,\bt) \+ \psi(z,\bt) \, \psi^*_x(z,\bt) \= b_x (z^2,\bt) \,.\nn
\end{align}
Solving this system we obtain
$$
\psi^*(z,\bt) \= \frac{b(z^2,\bt)}{\psi(z,\bt)} \,, \quad \psi_x(z,\bt) \= \psi(z,\bt) \, \frac{b_x(z^2,\bt)+2z}{2 \, b(z^2,\bt)} \,, 
\quad \psi_x^*(z,\bt) \= \frac{b_x(z^2,\bt)-2z}{2\, \psi(z,\bt)} \,. 
$$
Substituting these expressions into~\eqref{Dzwtdefinition}, we obtain the first equality of~\eqref{KDequation1}. 
Proposition~\ref{KDrelation} is then proved.
Theorem~\ref{thmmainabs} follows easily from Theorem~\ref{wavethm} and the first equality of~\eqref{KDequation1}, as 
the factors of the form $\psi(z,\bt)/\psi(w,\bt)$ cancel in each product of the sum of the right hand side of~\eqref{newformula}.  
\epf

As the remark given in the introduction (right after Theorem~\ref{thmmainabs}), the abstract version of Theorem~\ref{thmmainabs}
follows immediately. 
We note that it is not difficult to give a direct proof of the abstract version of Theorem~\ref{thmmainabs} 
by using the definition of~$K$ (cf. equation~\eqref{defbigK}), Lemma~\ref{lemmanablaR} and Proposition~\ref{taulemma}; we leave it as an exercise for interested readers.

\section{Generating series of the generalized BGW correlators}\label{section4}
Denote by $\pi: \, \overline{\mathcal{M}}_{g,n+1} \rightarrow \overline{\mathcal{M}}_{g,n}$ the forgetful map forgetting the last marked point, and 
$\rho:\overline{\mathcal{M}}_{g-1,n+2}\rightarrow \overline{\mathcal{M}}_{g,n}$ and  
$\phi_{h,I}: \overline{\mathcal{M}}_{h,|I|+1} \times \overline{\mathcal{M}}_{g-h,|J|+1} \rightarrow \overline{\mathcal{M}}_{g,n}$, $I \sqcup J=\{1,\dots,n\}$ the gluing maps.  
Norbury \cite{N} introduced a collection of cohomology classes 
$\bigl\{\Theta_{g,n}\in H^*\bigl(\overline{\mathcal{M}}_{g,n}\bigr)\bigr\}_{2g-2+n>0}$ satisfying 
\begin{align} 
& {\rm i)} ~ \Theta_{g,n} \mbox{ is of pure degree}\,, \label{np1} \\
& {\rm ii)} ~ \rho^*\Theta_{g,n} \=\Theta_{g-1,n+2} \,, \qquad \phi_{h,I}^*\Theta_{g,n} \= \pi_1^*\Theta_{h,|I|+1} \cdot \pi_2^*\Theta_{g-h,|J|+1} \,, \label{np2} \\
& {\rm iii)} ~ \Theta_{g,n+1} \= \psi_{n+1} \cdot \pi^*\Theta_{g,n} \,, \label{np3} \\
& {\rm iv)} ~ \Theta_{1,1} \= 3 \, \psi_1 \, .  \label{np4}
\end{align}
Norbury proved that  such $\Theta_{g,n}$ exists and it must satisfy $\Theta_{g,n}\in H^{4g-4+2n}\bigl(\overline{\mathcal{M}}_{g,n}\bigr)$. 
Define $\Ztheta = \Ztheta(\bt)$ as the following generating series of intersection numbers (called the partition function)
\beq
\Ztheta(\bt) 
\= \exp \biggl(\sum_{g,n\geq 0}\frac1{n!} \sum_{p_1,\dots,p_n\geq 0} 
\int_{\overline{\mathcal{M}}_{g,n}} \Theta_{g,n} \, \psi_1^{p_1} \dots \psi_n^{p_n} \, t_{p_1}\dots t_{p_n}\biggr) \,.
\eeq
Define $\utheta=\utheta(\bt):= \frac{\p^2 \log \Ztheta(\bt)}{\p t_0^2}$.  
The integrals $\int_{\overline{\mathcal{M}}_{g,n}} \Theta_{g,n} \, \psi_1^{p_1} \cdots \psi_n^{p_n}$, called 
the $n$-point $\Theta$-class intersection numbers (aka the 
$n$-point $\Theta$-class correlators), are independent of choice of~$\Theta_{g,n}$. They vanish unless 
$p_1+\dots+p_n=g-1$. Therefore $\utheta$ belongs to $\QQ\llm t_0,t_1,t_2,\dots \rrm$. 

\medskip

\noindent {\bf Norbury's Theorem} (\cite{N}).
\textit{The formal power series $\utheta$ is a solution of the KdV hierarchy~\eqref{KdVhk} with 
\beq\label{bgwin} \utheta(t_0=x,\bdzero) \= \frac{1}{8 \, (x-1)^2}\,.\eeq
Moreover, $\Ztheta$ is the tau-function of the solution $\utheta$ satisfying the string type equation}
\beq\label{stringbgw}
\sum_{i\geq 0} (1+2i) \, t_i \, \frac{\p \Ztheta}{\p t_i} \+ \frac 18 \Ztheta\= \frac{\p\Ztheta}{\p t_0}\,. 
\eeq

\smallskip

\noindent Norbury's Theorem tells that $\utheta=\uteight$. 
 The latter is defined in Example~\ref{example2} of Introduction.
So $\taueight$ and $\Ztheta$ can only differ by the exponential of a linear function. 
Noting that equation~\eqref{stringbgw} coincides with equation~\eqref{scaling} with $C=1/8$, so the linear function can only be a constant, which can easily 
be normalized as zero. We conclude 
that $\Ztheta=\taueight$. 
The latter is often called the BGW tau-function~\cite{BG,GW,A,DN,MMS}, originally studied in matrix models. 
The goal of this section is to prove Theorems~\ref{BGWnpoint},~\ref{RCC},~\ref{dbessel}, and give some explicit computations.
 
\subsection{The essential second kind topological ODE of $A_1$-type}
The topological ODE of $\g$-type with $\g$ being a simple Lie algebra was introduced and studied in~\cite{BDY2,BDY3}, which will be 
called the first kind topological ODE. 
To prove Theorem~\ref{RCC}, let us introduce an ODE associated to~${\rm sl}_2(\CC)$
\beq\label{diffrho}
2 \, \z^3 \, \rho''' \+ 3 \, \z^2 \rho'' \, -  \, 2\, \z\, \bigl(\z - 2 \,C\bigr)\, \rho' \,- \, 2 \, C\, \rho \= 0\,, \qquad '\:=\frac{d}{d\z}\,,
\eeq
where~$C$ is an arbitrary parameter. 
It will be used for computing the $\Theta$-class intersection numbers in full genera ($C=1/8$) 
and the generalized BGW correlators. 
As one motivation, note that the matrix~$R=R(\lambda,x)$ appearing in Theorem~\ref{RCC} is 
uniquely determined by the given $f(x)=\frac{C}{(1-x)^2}$, 
and equation~\eqref{diffrho} is nothing but a kind of transformation of 
the matrix resolvent recursive relation (see Section~\ref{s22}) for the corresponding~$b=b(\lambda,x)$, which has two 
independent variables, into a recursive relation for the coefficients of the one-variable function~$\rho$.
We call~\eqref{diffrho} the {\it essential second kind topological ODE of $A_1$-type}.  
See Section~\ref{further2} for more details.

\begin{prop} \label{bcc} For any fixed $C\in\mathbb{C}$, there exists a 
unique series~$\rho$ in $\CC\bigl[\!\bigl[z^{-1}\bigr]\!\bigr]$ 
satisfying equation~\eqref{diffrho} as well as the initial condition 
\beq\label{inirho}
\rho(\infty) \= 1\,.
\eeq
Moreover, $\rho\in \QQ[C]\bigl[\!\bigl[\z^{-1}\bigr]\!\bigr]$ and it satisfies the following nonlinear ODE 
\beq\label{nonlinearODE}
\z\, \rho'^2 \+ \Bigl(1-\frac{2C}{\z}\Bigr)\, \rho^2 \,-\, \rho\, \rho' \,-\, 2 \,\z\, \rho \,\rho'' \= 1 \,. 
\eeq
\end{prop}
\begin{proof} Write 
\beq\label{rhok}
\rho \=  \sum_{k\geq -1} (2k+1)!!  \,  \frac{\rho_k}{\z^{k+1}} 
\eeq
with $\rho_{-1}=1$. Substituting this expression in \eqref{diffrho} we obtain
\beq\label{rrec}
 \rho_k \= \frac{C+\frac{k(k+1)}2}{k+1} \, \rho_{k-1} \,, \qquad k\geq 0 \,. 
\eeq
This implies the uniqueness statement of the proposition and $\rho\in  \QQ[C]\bigl[\!\bigl[\z^{-1}\bigr]\!\bigr]$. To show \eqref{nonlinearODE}, note that 
$$
\frac{d}{d\z} \Bigl(\z\, \rho'^2 \+ \bigl(1-2 \, C/\z\bigr)\, \rho^2 
\,-\, \rho\, \rho' \,-\, 2 \,\z\, \rho \, \rho''
\Bigr) \= 0 \,. 
$$
Therefore, 
$\z\, \rho'^2 + \bigl(1-2 \, C/\z\bigr)\, \rho^2 - \rho\, \rho' - 2 \,\z\, \rho \, \rho'' \equiv C_1\,,$ 
where $C_1$ is a constant independent of $\z$. 
The fact that $C_1\equiv1$ can be deduced from \eqref{rhok}. The proposition is proved.
\end{proof}
We now apply Proposition~\ref{bcc} to derive the explicit expression of the basic matrix 
resolvent of the solution of the KdV hierarchy characterized by the initial data as in Example~\ref{example2}.

\medskip

\noindent {\it Proof} of Theorem~\ref{RCC}. Let $\rho=\rho(\z)$ be  
the unique element in Proposition~\ref{bcc}. According to Proposition~\ref{bcc}, $\rho=\rho(\z)$ satisfies equations  \eqref{diffrho}--\eqref{nonlinearODE}. 
Define $\tilde b(\lambda, x) = \rho\bigl(\lambda\, (x-1)^2\bigr)$. 
Then it is easy to check that $\tilde b(\lambda,x)$ satisfies \eqref{beq} and \eqref{essential} with $u$ replaced by $\frac{C}{(x-1)^2}$ and $\p$ by $\p_x$, 
and it has the form \eqref{defibk}. Hence the uniqueness statement in the definition of the basic matrix resolvent (see Section \ref{s22}) implies $b(\lambda,x) =\tilde b(\lambda,x)$. 
The above equation~\eqref{rrec} yields an explicit expression for~$\rho$:
\beq\label{expressionrho}
\rho \= 1 \+ \sum_{k\geq 0} \frac{(2k+1)!!}{(k+1)!}  \frac{\prod_{i=0}^k \bigl(C+\frac{i(i+1)}{2}\bigr)}{\z^{k+1}} 
\= {}_{3} F_0\biggl(\frac12, \,\frac12+\alpha, \, \frac12-\alpha; \, ; \frac1\z\biggr) \,.
\eeq
(Recall that $2\alpha=\sqrt{1-8C}$.)
This gives the $(1,2)$-entry of $\R$. Other entries can be obtained by using \eqref{aeq} and \eqref{ceq}.  \epf

As already pointed out after Theorem~\ref{RCC}, if $-C$ is a triangular number, i.e., $C=-\frac{p(p+1)}2$ for some $p\in \NN$, then the 
formal series $\rho=\rho(\z)$ truncates to the polynomial $\rho = {}_{3} F_0\bigl(\frac12, \,1+p, \, -p; \,; \frac1\z\bigr)$ in~$\z^{-1}$. 
This corresponds to the rational limit of the $p$-soliton solution to the KdV hierarchy.

For $n\geq 2$, using Corollary~\ref{cor2}  one immediately obtains an explicit formula of the generating series  of the  $n$-point 
generalized BGW correlators (any $C$) in terms of the matrix~$\R$ given in Theorem~\ref{RCC}.

\smallskip

\smallskip

\noindent {\it Proof} of Theorem~\ref{BGWnpoint}. 
Following from Norbury's Theorem, Corollary~\ref{cor2}, Theorem~\ref{RCC} with $C=1/8$. \epf

\smallskip

\noindent {\it Proof} of Corollary~\ref{1pointgbgw}.
Note that equation~\eqref{scaling} implies 
$$
(1+2p)\,  \otc_p (0) 
\=  \otc_{0,p} (0)  \,, \qquad p\geq 0\,. 
$$
Formula~\eqref{1ptC} is then a consequence of~\eqref{expressionrho} (as it is an easy exercise that $\Omega_{0,k}=b_k/(2k+1)!!$). \epf

\noindent {\it Proof} of Corollary~\ref{1pointbgw}.
According to Norbury's Theorem, formula~\eqref{state1} is a special case of~\eqref{1ptC}. \epf

\subsection{Proof of Theorem~\ref{dbessel} and Proposition~\ref{identities}} It is straightforward to check that 
the function $W:=\sqrt{2\xi/\pi}\,e^z K_\alpha(\xi)$ with 
$\xi=z(1-x)$ satisfies the differential equation $L(W)=z^2\,W$, where $L=\p_x^2+2C/(1-x)^2$. Moreover, as $z\to\infty$ 
within an appropriate sector, the asymptotic behavior of~$W$ coincides with~$\psi$ defined 
in~\eqref{gbgwbessel}. Similarly, the analytic function $W^*:=\sqrt{2\pi\xi}\,e^{-z}\,I_\alpha(\xi)$ 
satisfies $L(W^*)=z^2\,W^*$ and has asymptotics coinciding with $\psi^*$.  It remains to show the 
zero-residue condition. This can be proved by noticing that 
\begin{align}
\psi(z,x) \, \psi^*(z,x) & \=   \sum_{k\geq 0}  \frac{a_k(\alpha)}{z^k(1-x)^k}  \sum_{\ell\geq 0} (-1)^\ell \frac{a_\ell(\alpha)}{z^\ell (1-x)^\ell}  \nn\\
 & \= {}_{3} F_0\biggl(\frac12, \,\frac12+\alpha, \, \frac12-\alpha; \, ; \frac1{z^2(1-x)^2}\biggr) \= b(z^2,x)\,. 
\end{align}
This proves that $\psi$ and $\psi^*$ are a pair of wave functions as claimed, and equation~\eqref{newformulaforgbgw}
then follows from Theorem~\ref{wavethm} with~$\bt=0$. This completes the proof of Theorem~\ref{dbessel}.
 
For Proposition~\ref{identities} we first observe that the recursion and boundary conditions for~$A_{mn}(\alpha)$ given
in~\eqref{amnrec} are obviously equivalent to either of the closed formulas in~\eqref{amnsimple}, so it suffices to
prove the former. For convenience we write simply $a_k$ and $A_{mn}$ for~$a_k(\alpha)$ and $A_{mn}(\alpha)$ and set
$$ A(z)\,:=\,\psi(z,0)\,=\,\sum_{k=0}^\infty\frac{a_k}{z^k}\,, \qquad 
  B(z)\,:=\,\psi_x(z,0)\,=\,z\+\sum_{k=0}^\infty\frac{b_k}{z^k}$$
with $b_k=a_{k+1}+ka_k$.  (Avoid from confusion with the notations $a_k,b_k$ in other sections.) One checks easily that ($A(z),B(z))$ satisfies the first-order differential system 
\beq\label{ABsystem} 
\theta_z \begin{pmatrix} A(z)\\ B(z)\end{pmatrix} \= \begin{pmatrix}  z & -1 \\ 2C-z^2 & 1+z \end{pmatrix}
\begin{pmatrix} A(z)\\ B(z) \end{pmatrix}\,, \eeq
where $2C=\frac14-\alpha^2$ and $\theta_z:=z\,\frac d{dz}\,$. (This implies a second-order differential equation for~$A$ 
equivalent to the equation~$LW=z^2W$ given above.)  Setting $y=-w$, we find the two generating functions
$$ \sum_{m,n\ge0}\frac{(m+n)(A_{m+1,n}-A_{m,n+1})}{z^m\,y^n}
 \= \bigl(\theta_z\+\theta_y\bigr)\biggl(\frac{A(z)B(y)-B(z)A(y)}{z+y}\biggr)$$
(this follows from the definition of the $A_{mn}$ and the fact that the Euler operator $\theta_z+\theta_y$ 
annihilates the degree zero function $\frac{z-y}{z+y}$) and
$$ \sum_{m,n\ge0}\frac{(n-m)a_ma_n}{z^m\,y^n}    \= \bigl(\theta_z\,-\,\theta_y\bigr)\bigl(A(z)\,A(y)\bigr)\,.$$
Using the differential equations~\eqref{ABsystem} we find that the right hand sides of both of these formulas
equal $(z-y)A(z)A(y)+A(z)B(y)-B(z)A(y)$, proving equation~\eqref{amnrec}. For~\eqref{amnclosed}, 
we notice that the first of equations~\eqref{amnsimple} shows that the polynomial $A_{mn}(\alpha)$ is
divisible by $a_m(\alpha)$ (because each $a_r$ with $r\ge m$ is divisible by~$a_m$), and since $a_n(\alpha)$ has
only simple zeros, this shows that the quotient $\wt A_{mn}:=A_{mn}/a_ma_n$ has only simple poles at half-integral
values of~$\alpha$. Since $\wt A_{mn}$ is also small at infinity, it has a partial fraction decomposition as
$\sum_kc_{mn}(k)/(\alpha-k-\frac12)$ for some coefficients $c_{mn}(k)$.  That these coefficients have the values
given in~\eqref{amnclosed} can be proved by comparing residues, using either the recursion~\eqref{amnrec} or one
of the closed formulas~\eqref{amnsimple}, together with a simple binomial coefficient identity. The details are
left to the reader. 
\epf 

\subsection{Computations} 
\subsubsection{Some $\Theta$-class intersection numbers}
According to Norbury's Theorem, we have
$$\oteight_{p_1,\dots,p_n}(0) \= \sum_{g=0}^\infty \int_{\overline{\mathcal{M}}_{g,n}} \Theta_{g,n} \, \psi_1^{p_1} \dots \psi_n^{p_n} \, .$$ 
Here $n\geq 1$. 
Note that the degree-dimension matching reads 
$$
p_1+\dots+p_n + 2g-2+n \= 3g-3 + n ~\quad \Leftrightarrow  ~\quad  p_1+\dots+p_n \=g-1 \,.
$$
So, actually, $\oteight_{p_1,\dots,p_n}(0)
=  \int_{\overline{\mathcal{M}}_{1+p_1+\dots+p_n,n}} \Theta_{1+p_1+\dots+p_n,n} \, \psi_1^{p_1} \cdots \psi_n^{p_n}
$ as we have already given in~\eqref{Norburyci}. We are going to compute some $\Theta$-class intersection numbers of the form 
$\int_{\overline{\mathcal{M}}_{1+nb,n}} \Theta_{1+nb,n} \, \psi_1^{b} \cdots \psi_n^{b}$, denoted by 
$\oteight_{b^n}(0)$ for short. 
Using an algorithm designed in~\cite{DuY1} and using Theorem~\ref{BGWnpoint}, one can compute these 
$\Theta$-class intersection numbers in relatively high genera. For example we have 
$$\oteight_{1^{11}} (0) \=  \frac{3727154672771403705644393643825}{67108864}\,,$$ 
$$\oteight_{2^8} (0) \=  \frac{10497097022517857530944569189202112366380675}{36028797018963968}\,,$$
$$\oteight_{3^6} (0) \=  \frac{291143373745168297982109927833062542748609508458550221626925092375}{604462909807314587353088}\,.$$
Several more $\Theta$-class correlators are in Table~\ref{table1}.
\begin{table}[!htbp]
\begin{center}\tiny 
    \begin{tabular} {|c|M{0.9cm}|M{1.9cm}|M{3.8cm}|M{5.6cm}|N}
    \hline
          & $b=0$ & $b=1$ & $b=2$ & $b=3$ &  \\[8pt]
    \hline
    $n=1$ &  $\frac18$ & $\frac3{128}$ & $\frac{15}{1024}$ &  $\frac{525}{32768}$  &   \\[8pt]
    \hline
    $n=2$ & $\frac18$ & $\frac{63}{512}$ & $\frac{125565}{131072}$  & $\frac{178066035}{8388608}$ &    \\[8pt]
    \hline
    $n=3$ & $\frac14$  & $\frac{7221}{2048}$ &  $\frac{8160299505}{8388608}$ & $\frac{5357097499513095}{4294967296}$  &   \\[8pt]
    \hline
    $n=4$ & $\frac34$    & $\frac{4825971}{16384}$  & $\frac{6118287865593075}{1073741824}$  
    &  $\frac{3673662570422147820860595}{4398046511104}$  &   \\[8pt]
    \hline
    $n=5$ & 3   &  $\frac{3540311739}{65536}$  &  $\frac{2089963670900974355205}{17179869184}$ &  
    $\frac{7614423907504732590945890803999875}{2251799813685248}$  &  \\[8pt]
    \hline
    $n=6$ & 15 & $\frac{1209901485555}{65536}$ &  $\frac{31867458860062839143669852025}{4398046511104}$ 
    &  $\frac{32942281960173069977596091564715863342175375}{576460752303423488}$ &  \\[8pt]
    \hline
    \end{tabular}
\end{center}
\caption{$\oteight_{b^n} $,\quad$b=0,1,2,3$\,.} \label{table1}
\end{table}
One observes that $\oteight_{0^n} (0) = \frac{(n-1)!}{8}$.

The algorithm designed in~\cite{DuY1} also
 produces explicit full genera formulas for the $\Theta$-class correlators: 
\begin{align}
&  \sum_{p\geq 0} \frac{(1+2p)!!}{\lambda^{p+1}} \oteight_{0,p} (0) \= b-1 \nn\\
&  - 512\, 3!! \, 5!!\sum_{i\geq 0} \frac{(1+2p)!!}{\lambda^{p+1}} \oteight_{1,2,p} (0)  \=   
-\,1024 \, \lambda^3 \,  a  \+ 256  \, \lambda^2 \, ( b - a ) \+ \lambda \, \bigl( 848   b - 616   a  - 128   c\bigr) \nn\\
&\qquad \qquad \qquad \qquad \qquad \qquad \qquad  \+ 3321  b  \,-\, 2187  a \, - \,432  c \nn\\ 
&  1024\, 3!!^3 \sum_{p\geq 0} \frac{(1+2p)!!}{\lambda^{p+1}} \oteight_{1,1,1,p} (0) \= 
2048 \, \lambda^4 b \+ \lambda^3 \, \bigl(256   b-2048  c\bigr) \,-\, \lambda^2 \,\bigl(1536  a +96   b+768  c \bigr) \nn\\ 
& \qquad \qquad \qquad \qquad \qquad \quad  
\+ \lambda \, \bigl(1728    a - 8820   b - 864   c  \bigr) \+ 26352  a \,- \, 46989  b \+ 4284  c   \nn
\end{align}
In the above formulas, $a = \frac{1} { 8 \lambda}  \,  {}_{3} F_0\Bigl(\frac32, \, \frac12, \, \frac12; \, ; \frac{1}{\lambda}\Bigr)$,  $b = {}_{3} F_0\Bigl(\frac12, \, \frac12, \, \frac12; \, ; \frac{1}{\lambda}\Bigr)$, and 
$$ c \=  \Bigl( \lambda-\frac14 \Bigr) \, 
{}_{3} F_0\biggl(\frac12, \, \frac12, \, \frac12; \, ; \frac{1}{\lambda}\biggr)
 \,-\, \frac38 \frac{1}{\lambda} \, {}_{3} F_0\biggl(\frac32, \, \frac12, \, \frac12; \, ; \frac{1}{\lambda}\biggr) \, - \, 
 \frac{27}{32} \frac{1}{\lambda^2} \, {}_{3} F_0\biggl(\frac52, \, \frac12, \, \frac12; \, ; \frac{1}{\lambda}\biggr)   \,. $$

\subsubsection{Correlators with parameter $C$} We list some 
correlators of the form $\otc_{b^n}(0)$ in Tables \ref{table2}--\ref{table3}.  
\begin{table}[!htbp]
\begin{center}\tiny 
    \begin{tabular} {|c|M{0.8cm}|M{10.5cm}|N}
    \hline
    & $b=0$ & $b=1$ &    \\[8pt]
    \hline
    $n=1$ & $C$ & $ \frac{1}6 C(C+1)$ &      \\[8pt]
    \hline
    $n=2$ & $C$ & $\frac16 C(C+1)(2C+5)$ &   \\[8pt]
    \hline
    $n=3$ & $2C$  & $\frac13 C(C+1)(2C+7)(3C+10)$ &    \\[8pt]
    \hline
    $n=4$ & $6C$    & $ C(C+1)(22 C^3 +292 C^2+1320 C+1925)$  &   \\[8pt]
    \hline
    $n=5$ &  $24C$  &  $ C(C+1)(364 C^4+8028 C^3+69089 C^2+261625 C+350350)$  &     \\[8pt]
    \hline
    $n=6$ & $120C$ & $ C(C+1)(8160 C^5+272480 C^4+3843730 C^3+27340910 C^2+93831500 C+119119000)$ &     \\[8pt]
    \hline
    \end{tabular}
\end{center}
\caption{$\otc_{b^n}(0)$, \quad $b=0,1$\,.} \label{table2}
\end{table}

\begin{table}[!htbp]
\begin{center}\tiny 
    \begin{tabular} {|c|M{15cm}|N}
    \hline
    $n=1$ & $\frac{1}{30}C(C+1)(C+3)$ &       \\[8pt]
    \hline
    $n=2$ & $ \frac{1}{60}C(C+1)(C+3)(3 C^2+38 C+126)$ &    \\[8pt]
    \hline
    $n=3$ & $\frac{1}{60}C(C+1)(C+3)(15 C^4+550 C^3+8011 C^2+52521 C+126126)$  &    \\[8pt]
    \hline
    $n=4$ & $\frac{1}{120}C(C+1)(C+3)(285 C^6+21120 C^5+701455 C^4+12823420 C^3+131525532 C^2+698301072 C+1466593128)$    &    \\[8pt]
    \hline
    $n=5$ &  $\frac{1}{30}C(C+1)(C+3)(1035 C^8+131220 C^7+7929500 C^6+286890460 C^5+6581287505 C^4+95511193020 C^3+838324176858 C^2+3995785717308 C+7792009289064)$  &     \\[8pt]
    \hline
    $n=6$ & $\frac{1}{72}C(C+1)(C+3)(49329 C^{10}+9683190 C^9+941488056 C^8+57517664804 C^7+2370818604241 C^6+67214920642718 C^5+1301029520426886 C^4+16691838842700000 C^3+133819015248860760 C^2+596987475819494760 C+1110408075747354384)$ &     \\[8pt]
    \hline
    \end{tabular}
\end{center}
\caption{$\otc_{2^n}(0) $\,.} \label{table3}
\end{table}

We list first few explicit formulas for certain correlators in full genera:
\begin{align}
& \sum_{i\geq 0} \frac{(1+2p)!!}{\lambda^{p+1}} \, \otc_{0,p}(0)   \= b-1 \nn\\
& 3!! \sum_{i\geq 0} \frac{(1+2p)!!}{\lambda^{p+1}}  \, \otc_{1,p}(0)  \= \lambda\, (2  b -3)+c -C \, b \nn\\
& 5!! \sum_{i\geq 0} \frac{(1+2p)!!}{\lambda^{i+1}} \, \otc_{2,p}(0)  
\= \lambda^2(3b-5) \+ 2 \, \lambda\,(c - C \, b) \+ 2 C \, a - \frac{C(C+3)}2 \,  b  \+ C \, c \nn\\
& -  \,  \frac{3!!^2}2\sum_{p\geq 0} \frac{(1+2p)!!}{\lambda^{i+1}} \, \otc_{1,1,p}(0) 
\= -\, \lambda^2 \, a \+ 2 \, C\, \lambda \, (   b - a  )  \,-\, C (C+3) \,   a + C(2C+3) \,  b - C\, c  \nn\\
& -  \, 2\,5!!^2 \sum_{i\geq 0} \frac{(1+2p)!!}{\lambda^{p+1}} \, \otc_{2,2,p}(0)  \= -4 \, \lambda^4 a \+ 8 \, C \lambda^3 (  b -  a )
\+ 4\, C \lambda^2\bigl((8C+9) \, b  -  c  - (4C+6) \, a  
\bigr) \nn\\
& \qquad   \+ C \lambda \, \bigl( (36 C^2 + 216 C + 180) \,   b - (12 C^2  + 132 C  + 120) \,  a  - 24 (C +1) \,  c\bigr) \nn\\
& \qquad   \+ 12C\bigl( 3 C^3  + 49 C^2  + 151 C +105\bigr) \, b  \,-\, 18 C \bigl(5 + 7 C +2 C^2\bigr) \, c \,- 9 C\bigl( C^3 + 30 C^2  + 99 C + 70\bigr) \,   a  \nn
\end{align}
with  $a=\frac {C} {\lambda}  G_{\frac32}(\lambda)$,  $b=G_{\frac12}(\lambda)$, and 
$c \=(\lambda-2C) \, G_{\frac12}(\lambda)
- \frac{3C}{\lambda} \, G_{\frac32}(\lambda) - \frac{6C(C+1)}{\lambda^2} G_{\frac52}(\lambda)  $, where 
we recall that $G_{\alpha}(\z)$ is defined in~\eqref{defG}.

\section{Generating series of the Lam\'e partial correlation functions} \label{section5}
In this section, we study the Lam\'e tau-function $\tau_{{\rm elliptic}}$. 
\subsection{An explicit recursion for the basic matrix resolvent} Recall that the Lam\'e solution $u_{\rm elliptic}$ of the KdV 
hierarchy is the unique solution in the ring $V\llm \bt_{>0} \rrm$, where $V=\CC[g_2,g_3,\wp,\wp']/(\wp'^2- 4 \wp^3+g_2 \wp + g_3)$, 
having the initial data $f=C\,\wp(x;\tau)$. Let $\R(\lambda,x)$ and $b(\lambda,x)$ denote the functions as in Section~\ref{Introapp} (Example~3). 
Replacing $\p$ with $\p_x$ and with $u$ with $f=C \, \wp(x; \tau)$ in~\eqref{essential} we obtain 
the following nonlinear ODE for $b(\lambda,x)$ 
\beq\label{B10}
b_{xx} \, b  \m  \frac{1}{2} \, b_x^2 \,-\,  2 \, \bigl(\lambda- 2 \,  C \wp(x; \tau)\bigr)\,b^2 \= - \,2\, \lambda \, .
\eeq
We note again that $b=b(\lambda,x)$ is the unique power series solution in~$\lambda^{-1}$, having the form
$b= \sum_{k\geq-1} b_k/\lambda^{k+1}$ with $b_{-1}=1$, to equation~\eqref{B10}. 
Denote $X=  \wp(x; \tau),$ $Y=  \wp'(x; \tau)$. 
Recall that 
\[Y^2 \= 4 X^3 -g_2\, X- g_3 \= 4 \, (X-e_1)(X-e_2)(X-e_3)  \]
with $g_2$ and~$g_3$ as in Section~\ref{section11}.
In the $(X,\lambda)$-coordinates, equation~\eqref{B10} has the equivalent expression
\beq\label{B2} 
\Bigl(X^3 -\frac{g_2}4 X - \frac{g_3}4\Bigr)(2 b b'' - b'^2)  \+ \Bigl(3 X^2 - \frac{g_2}4\Bigr) b b' \+ (2CX-\lambda) b^2 \+ \lambda \= 0\,.
\eeq
Here $'=\frac{\p}{\p X}$.
It follows that $b_k$ must have the form $b_k=P_{k+1}$ with 
$P_0=1$ and $P_m$~$(m\geq 1)$ being polynomials in $X$ of degree $m$ whose coefficients are polynomials of $g_2,g_3,C$. 
Equation~\eqref{B2} uniquely determines $P_k,~k\geq 1$ in the following recursive way:
\begin{align}
P_k \=  & \sum_{i=0}^{k-1} \biggl[C\,  X P_i P_{k-i-1} +  \Bigl(\frac32 X^2- g_2\Bigr) P_i P_{k-i-1}' 
                  + \bigl(X^3 - \frac14 g_2 X- \frac14 g_3\bigr) \Bigl(P_{i} P_{k-i-1}'' - \frac12 P_i' P_{k-i-1}' \Bigr)\biggr] \nn\\
              &   \, - \, \frac12 \sum_{i=1}^{k-1} P_i P_{k-i} \,.
\end{align}
We do not have a closed expression of~$b$ for a general value of~$C$.

\subsection{Proof of Theorem~\ref{main2}} Define 
$\R^{\rm sp}(\lambda, x):=  \sqrt{\frac{S_p(\lambda)}{\lambda}} \, \R(\lambda,x)$. We have 
\begin{align}
& \bigl[\cL\,,\,\R^{\rm sp}\bigr] \= 0 \,, \nn\\
& \det \R^{\rm sp} \= - \, S_p(\lambda)\,. \nn
\end{align}
One identifies $\R^{\rm sp}$ with the matrix resolvent in~\cite{Du2}. This shows that 
$\sqrt{\frac{S_p(\lambda)}{\lambda}} \, b(\lambda,x)$ is a polynomial in~$\lambda$ of degree~$p$ with leading coefficient~1. 
The fact that $\sqrt{\frac{S_p(\lambda)}{\lambda}} \, b(\lambda,x)$ is also a degree~$p$ polynomial of $X=\wp$  can be 
deduced from homogeneity. Indeed, if we introduce a gradation by ${\rm wt} \, \lambda=2$, ${\rm wt} \, X=2$, 
${\rm wt} \, g_2 = 4$, ${\rm wt} \, g_3 = 6$, then we find that $b(\lambda,x)$ is homogeneous of degree 0. Note that 
$S_p(\lambda)$ is homogeneous of degree $4p+2$, i.e. ${\rm wt} \, S_p(\lambda) = 4p+2$. The theorem is proved. \epf

\subsection{Computations}  
The first few partial correlation functions of~$u_{\rm elliptic}$ are 
\begin{align}
& \Omega_{0,0} (x) \= C\, \wp(x) \,, \quad \Omega_{0,0,0} (x) \= C\, \wp'(x) \,, \quad 
\Omega_{0,0,0,0} (x) \= 6 \, C\, \wp(x)^2 \,-\, \frac{C\, g_2}2 \,,  \nn\\
& \Omega_{1,1} (x) \= C \, \biggl( \frac{(C+1) (2 C+5)}6 \wp(x)^3 -\frac{(C+1)  g_2  }{8} \wp(x)-\frac{C+2}{24}  g_3 \biggr) \,, \nn\\
& \Omega_{1,1,1} (x) \= \frac{C(C+1)}{24}  \,  \wp'(x) \,  \Bigl( 4 (2 C+7) (3 C+10) \wp(x)^3 - (11 C+28)  g_2  \wp(x) - 2 (C+5)  g_3 \Bigr) \,.  \nn
\end{align}
Considering the Laurent series expansion of~$\wp(x)$ at $x=0$, we obtain
\begin{align}
& \Omega_{0,0} (x)  \= 
  C\, \Bigl(\frac1{x^2} + \frac{g_2}{20}\,x^2 +\frac{g_3}{28}\,x^4 +\frac{g_2^2}{1200}\,x^6 +\frac{3 g_2 g_3}{6160}\,x^8 
+ \frac{49g_2^3+750g_3^2}{7644000}\,x^{10} + \cdots \Bigr) \,, \nn\\
& \Omega_{0,0,0}(x)\= 
 C\, \Bigl(\frac2{x^3} - \frac{g_2}{10}\,x -\frac{g_3}{7}\,x^3 -\frac{g_2^2}{200}\,x^5 -\frac{3g_2 g_3}{770}\,x^7
- \frac{49g_2^3+750g_3^2}{764400}\,x^9 - \cdots \Bigr) \,, \nn\\
& \Omega_{1,1} (x) \=  \frac{C(C+1) (2 C+5)}{6\,x^6}  \+  \frac{C^2(C+1) }{20\,x^2}g_2 
\+  \frac{C(6C^2 + 14 C +1)}{168} g_3  \nn\\
& \qquad \qquad \qquad  \+  \frac{C(C+1) (8 C+5) }{2400} g_2^2x^2 \+ \cdots \,, \nn\\
& \Omega_{1,1,1} (x) \=   \frac{C (C+1)(2C+7)(3C+10)}{6\,x^9}   \+  \frac{C (C+1) (9 C^2+34 C+35)}{60\,x^5}g_2 \nn\\
& \qquad \qquad \qquad \+  \frac{C (C+1) (18 C^2+109 C+140)}{168\,x^3}  g_3 \+  \frac{C (C+1) (24 C^2+109 C+140)}{2400\,x} g_2^2 \+ \cdots \, .  \nn
\end{align}
Since each of these expressions is modular of some weight with respect to the action $\tau\mapsto\frac{a\tau+b}{c\tau+d}$,
$x\mapsto\frac x{c\tau+d}$ of the full modular group, we call these Lam\'e partial correlation functions {\it modular deformations} of the 
generalized BGW partial correlation functions (after a shift of~$x$ by~1). 
The partial correlation functions also allow us to consider other analytic in~$\tau$ 
aspects.
For example, using  
\[
\wp(x) \= \frac{\pi^2}3 \bigl(2+3\cot^2 (\pi x)\bigr) \+ 16 \, q^2 \, \pi^2 \,  \sin^2 (\pi x)  \+  4 \, q^4 \bigl(5+4\cos(2\pi x)\bigr) \sin^2 (\pi x) \+ \cdots\,
\]
one can obtain the $q\rightarrow 0$ expansion ($q=e^{\pi i \tau}$) of 
the Lam\'e partial correlation functions.  
The above modular deformations could be interpreted in an alternative way.
Namely, we first switch on the periods $(2\omega,2\omega')$ 
of the Weierstrass $\wp$-function. (Previously we have taken $\omega=\frac12$ and $\omega'=\frac{\tau}2$). It is easy to see that the 
form of the partial correlation functions do not change when the initial data is given by $f(x)\equiv C \, \wp(x; 2\omega,2\omega')$. Denote $\tau=\frac{\omega'}\omega$, $q=e^{\pi i \tau}$. Then 
we have as $\omega\rightarrow \infty$,
$$
\wp (x) \= \frac1{x^2} 
\+ \frac{\pi ^4 x^2 (1+240 q^2 +2160 q^4+\dots)}{240 \, \omega^4} \+ \frac{\pi ^6 x^4 (1 -504 q^2 -16632 q^4+\dots)}{6048 \, \omega^6} \+ O\bigl(\omega^{-7}\bigr) \,.
$$
Hence the modular deformation of a generalized BGW partial correlation function 
can also be viewed as an $\omega\rightarrow \infty$ 
asymptotic.
\begin{remark}
For $C=1/8$, we expect the existence of a deformation of the Norbury class such that the partition function is equal to the Lam\'e tau-function. 
For other values of~$C$, we also expect the existence of cohomology classes on $\overline{\mathcal{M}}_{g,n}$ giving rise to the Lam\'e tau-function.
\end{remark}

Using Corollary \ref{cor2} we can obtain formulas like 
\begin{align}
& -  \, \frac{3!!^2} 2\sum_{i\geq 0} \frac{(1+2i)!!}{\lambda^{i+1}} \, \Omega_{1,1,i}(x) 
\=  -\, \lambda^2 \, a \+ C\, \lambda \bigl(  \wp'(x) \, b - 2   \wp(x) \, a  \bigr) \nn\\
& \qquad \qquad  \+ \frac{C}{4} \, \bigl(g_2-4 (C+3) \wp(x)^2\bigr) \, a \+ \frac{C(2 C+3)}2  \,  \wp(x) \wp'(x) \, b \,-\, \frac{C}2   \wp'(x) \, c \,, 
\end{align}
where $a,b,c$ are the components of $\R(\lambda,x)$. 

\section{Further remarks} \label{futherrmk}

\subsection{The second kind topological ODE of~$A_1$-type} \label{further2} 
Consider the solution of the KdV hierarchy with initial data~$f=C/(1-x)^2$. 
Let $\R=\R(\lambda,x)$ denote the basic matrix resolvent of this solution evaluated at $\bt_{>0} = \bdzero$ and $t_0=x$. 
Define $M(\lambda,x) = \frac1{\lambda^{\frac12}}  \, \lambda^{\frac{\sigma_3}4} \R(\lambda,x) \, \lambda^{-\frac{\sigma_3}4}$. 
According to Theorem~\ref{RCC} we have $M(\lambda,x)=M(\z)$ with $\z=\lambda\, (x-1)^2$. Moreover, $M=M(\z)$ satisfies
\beq \label{thetaode}
\frac{dM}{d\z} \+ \frac12 \biggl[ \begin{pmatrix} 0  & 1/\sqrt{\z} \\ 1/\sqrt{\z} - 2C/\z & 0 \end{pmatrix} \,, \, M\biggr] \= 0 \,. 
\eeq
We call \eqref{thetaode} the {\it second kind topological ODE of $A_1$-type}. Note that as $\zeta\rightarrow\infty$, we have
$$
M(\z) \, \sim \, \begin{pmatrix} 0  & 1\\ 1 & 0 \end{pmatrix}  \+  {\rm O}\Bigl(1/\sqrt{\z}\Bigr)\,. 
$$
Equation~\eqref{thetaode} together with this boundary condition uniquely determines~$M$. This~$M$ further satisfies $\det M(\z)=-1$. 
The second kind topological ODE~\eqref{thetaode} can be written equivalently as 
\beq \label{thetaode1}
\frac{dM}{dz} \+  \biggl[ \begin{pmatrix} 0  & 1 \\ 1 - \frac{2C}{z^2} & 0 \end{pmatrix} \,, \, M\biggr] \= 0 \,,
\eeq
where $\z=z^2$. More details about the second kind topological ODE and its generalization to an arbitrary simple Lie algebra will be given in a subsequent publication. 

\subsection{$M$-bispectrality} \label{bispsection}
Let us consider $V=\CC\llb x \rrb$, the ring of Laurent series of~$x$. Let $u=u(\bt)$ be a solution 
of the KdV hierarchy~\eqref{KdVhk} in $V\bllm\bt_{>0}\brrm$, $f(x)=u(x,\bdzero)$ the initial value. 
Let us recall a notion of bispectral solutions~\cite{DYZ} defined by using $\R(\lambda,x)$. Denote $\sigma_3=\sm 1  0  0 {-1}$.
\begin{defi}
The solution~$u$ is called {\it $M$-bispectral} if there exist three scalar functions $g(\lambda)\not\equiv0$, $g_0(\lambda)\not\equiv0$, $h(\lambda,x)$, and an ${\rm sl}_2(\CC)$-valued function 
$\widetilde \R(h)$ such that  
\beq
 \R(\lambda,x)  \= g(\lambda) \, g_0(\lambda)^{\sigma_3} \, \widetilde \R(h(\lambda,x)) \, g_0(\lambda)^{-\sigma_3} \,.
\eeq
\end{defi}
\begin{theorem} \label{bispecthm} 
The solution $u$ of the KdV hierarchy is M-bispectral iff 
\beq\label{ubisp}
\mbox{either } \quad f(x) \=   (x-A) \, C\,, \quad \mbox{or} \quad f(x) \= \frac{C}{(x-A)^2} \+ B \,,
\eeq
where $A,B,C$ are constants.
\end{theorem}
\begin{proof}
Denote by $a(\lambda,x),b(\lambda,x),c(\lambda,x)$ the entries of $\R(\lambda,x)$. They must have the form
\begin{align}
& a(\lambda,x) \= g(\lambda) \, \rho_1(h(\lambda,x)) \,,\\
& b(\lambda,x) \= q(\lambda) \, g(\lambda) \, \rho_2(h(\lambda,x)) \,,\\
& c(\lambda,x) \= q(\lambda)^{-1} \, g(\lambda) \, \rho_3(h(\lambda,x)) \label{third1}
\end{align}
for some functions $g(\lambda) \not\equiv 0$, $q(\lambda) =g_0(\lambda)^2\not\equiv 0 $, $\rho_1(h)$, $\rho_2(h) \not\equiv 0$, $\rho_3(h)$.
Substituting these expressions in equation~\eqref{aeq}  we obtain
$$\frac12 \, q(\lambda) \, g(\lambda) \, \rho_2'(h(\lambda,x)) \, h_x \= g(\lambda) \, \rho_1(h(\lambda,x))\,. $$
This implies existence of three functions of one-variable $W_1,W_2,P$ satisfying 
\beq\label{hP}
h(\lambda,x) \= P\bigl( x\, W_1(\lambda) + W_2(\lambda) \bigr) \,. 
\eeq
Hence we can assume that 
\beq
h(\lambda,x) \= x\, W_1(\lambda) + W_2(\lambda)\,, \qquad W_1(\lambda) \, \not\equiv \, 0 \,.
\eeq
It further implies that 
$$\frac12 \, q(\lambda)  \, W_1(\lambda) \= \frac{\rho_1(h(\lambda,x))}{\rho_2'(h(\lambda,x))} \= C_0 \,, $$
where $C_0$ is a constant. Using equations~\eqref{third1} and~\eqref{ceq} we have 
$$
\bigl( \lambda-2\,f(x) \bigr) \, \rho_2\bigl(h(\lambda,x)\bigr)  \,-\, \frac{W_1(\lambda)^2}2  \, \rho_2''\bigl(h(\lambda,x)\bigr) 
\= \frac{W_1(\lambda)^2}{4 \, C_0^2} \, \rho_3\bigl(h(\lambda,x)\bigr)\,.
$$
Therefore
\beq\label{eq9}
\lambda\,-\,2 \, f(x) \=  \frac{W_1(\lambda)^2}2  \, \frac{\rho_2''\bigl(h(\lambda,x)\bigr)}{ \rho_2\bigl(h(\lambda,x)\bigr)  } \+ \frac{W_1(\lambda)^2}{4 \, C_0^2 } \, 
\frac{\rho_3\bigl(h(\lambda,x)\bigr)}{ \rho_2\bigl(h(\lambda,x)\bigr) } \,.
\eeq
Differentiating both sides w.r.t.~$\lambda$ we obtain
\beq
1 \= \p_\lambda \Biggl( \frac{W_1(\lambda)^2}2  \, \frac{\rho_2''\bigl(h(\lambda,x)\bigr)}{ \rho_2\bigl(h(\lambda,x)\bigr)  } \+ \frac{W_1(\lambda)^2}{4 \, C_0^2 } \, 
\frac{\rho_3\bigl(h(\lambda,x)\bigr)}{ \rho_2\bigl(h(\lambda,x)\bigr) }  \Biggr) \,.
\eeq
Denote 
$$
G(h)\:=  \frac12  \, \frac{\rho_2''(h)}{ \rho_2(h)  } \+ \frac{1}{4 \, C_0^2} \, \frac{\rho_3(h)}{ \rho_2(h) }  \,.
$$
We have 
$$
1 \= \p_\lambda\bigl(W_1(\lambda)^2 \, G(h(\lambda,x))\bigr) \= 2 \, W_1 \, W_1' \, G \+ \bigl(x\,W_1'+W_2'\bigr) \, W_1^2 \, G'  \,.
$$
Differentiating both sides with respect to~$x$ we find
$$
0  \= 3 \, W_1^2 \, W_1' \, G'   \+  \bigl(x\,W_1'+W_2'\bigr) \, W_1^3 \, G'' \,.
$$
This leads to two possibilities:
\begin{align}
&  \mbox{i) }  W_1' \=0 \,:   \quad W_2' \, G''\=0 \, ,\\
&  \mbox{ii) } W_1'\neq 0 \,:   \quad \biggl(G''\not\equiv 0 \mbox{ and }  x  =  - \, \frac{3  \, W_1' \, G'   \+  W_2'\, W_1 \, G''}{W_1' \, W_1 \, G''} \biggr) \quad \mbox{ or } \quad \bigl(G''\= G'\=0\bigr)\,.
\end{align}
For the case~i), we have $W_1\equiv C_1$ for some non-zero constant $C_1$. So 
$$
1 \= W_2' \, G'\, C_1^2 \,.
$$
This implies that $W_2'\neq 0$ and so $G''\equiv0$. Hence $G(h)=C_2 \, h+ C_4$ and $W_2=\frac{\lambda}{C_2 C_1^2} \+ C_3$ for some constants $C_2,C_3,C_4$.
Substituting these expressions in~\eqref{eq9} we find 
$$
u \= - \, \frac12 \, C_1^3 \, C_2 \, x \,-\, \frac12 \, C_1^2 \, C_2 \, C_3 \,-\,\frac12 \, C_1^2 \,  C_4 \,.
$$
For the case~ii), we have 
$$
1  \=  - \, 3\, \biggl( \frac{G'   }{G''} \biggr)'   \,, \quad \mbox{ or } \quad G\= {\rm Const}\,.
$$
So $$G(h)\=\frac{C_1}{(h+C_2)^2} \+ C_3$$
for some constants $C_1,C_2,C_3$. Then we find that 
\begin{align}
& W_1 \= \biggl(\frac{\lambda \, + \, 2 \, C_3 \, C_4}{C_3}\biggr)^{\frac12} \,,\\
& W_2 \=  C_5 \, \bigl(\lambda \,+ \, 2 \, C_3 \, C_4\bigr)^{\frac12} \, - \, C_2 \,,\\
& f \=  - \, \frac{C_1}{2 \bigl(x + \sqrt{C_3} \, C_5\bigr)^2} \, - \, C_3\, C_4 
\end{align}
for some constants $C_4$, $C_5$. We have proved that $u$ must have the form in~\eqref{ubisp}. 

We are left to prove the solvability of existence of $g_0,g,\widetilde R,h$ when $f$ is taken to be one of~\eqref{ubisp}. 
For $f(x)=(x-A)\, C$, we can take for instance 
$$g_0\=1\,, \quad g\=  \frac{\sqrt{\lambda}}{2^\frac13 C^\frac13}\,, \quad W_1\=  - \, 2^\frac13 C^\frac13\,, \quad W_2 \= \frac{\lambda}{2^\frac23 C^\frac23}\,,$$ 
and we have
$$
b \= \frac{\sqrt{\lambda}}{2^\frac13 C^\frac13}\, \sum_{g=0}^\infty \frac{(6g-1)!!}{96^g\,g!} \, \zeta^{-3g-\frac12} \, , \qquad \zeta \:= - \, (x-A) \, 2^\frac13 C^\frac13 \+ \frac{\lambda}{2^\frac23\, C^\frac23} \,.
$$
For $f(x) = \frac{C}{(x-A)^2} + B$, we can take for instance 
$$g_0 \= \frac1{(\lambda-2B)^{\frac14}}\,, \quad g \= \sqrt{\lambda}\,, \quad W_1 \= \sqrt{\lambda-2B}\,, \quad W_2 \= - \, A\, \sqrt{\lambda-2B}\,,$$ 
and we have 
$$
b \= \frac{\sqrt{\lambda}}{\sqrt{\lambda \,-\,2B}} \,  {}_{3} F_0\biggl(\frac12, \,\frac{1+\sqrt{1-8C}}2, \, \frac{1-\sqrt{1-8C}}2; \, ; \frac1{(x-A) \sqrt{\lambda-2B}}\biggr).
$$
The theorem is proved.
\end{proof}

\begin{remark}
For the first case of~\eqref{ubisp}, when $(A,C)$ is taken as $(0,1)$, the above theorem recovers the following
formula of~\cite{BDY1,BDY3,Zhou1}:
\begin{align}
&  \sum_{g,p_1,\dots,p_n\geq 0}  \int_{\overline{\M}_{g,n}} \psi_1^{p_1} \cdots  \psi_n^{p_n} \, 
\prod_{j=1}^n \frac{(2p_j+1)!!}{\lambda_j^{p_j+1}}   \=-\sum_{\sigma\in S_n/C_n} 
\frac{\tr \,M (\lambda_{\sigma(1)})\dots M(\lambda_{\sigma(n)})}
 {\prod_{i=1}^n \bigl(\lambda_{\sigma(i+1)}-\lambda_{\sigma(i)}\bigr)} 
\,-\, \delta_{n2} \frac{\lambda_1+\lambda_2}{(\lambda_1-\lambda_2)^2}\,, \label{npointwkold}
\end{align}
where $M(\lambda)$ is a 2 by 2 matrix given explicitly by
\beq \label{MWK}
M(\lambda) \= \begin{pmatrix}
\frac12\sum_{g=1}^\infty \frac{(6g-5)!!}{24^{g-1}\, (g-1)!} \lambda^{-3g+2} & \sum_{g=0}^\infty \frac{(6g-1)!!}{24^g\, g!} \lambda^{-3g}\\
\\
 \sum_{g=0}^\infty  \frac{1+6g}{1-6g} \frac{(6g-1)!!}{24^g\, g!} \lambda^{-3g+1} & - \frac12\sum_{g=1}^\infty  \frac{(6g-5)!!}{24^{g-1}\, (g-1)!} 
\lambda^{-3g+2}\\
\end{pmatrix} \,.
\eeq
We note that~$\lambda^{-\frac12} M(\lambda)$ satisfies an ODE in~$\lambda$, {\it the first 
kind topological ODE of $A_1$-type}~\cite{BDY2,BDY3}. 
\end{remark}

\begin{remark}
For the first case of~\eqref{ubisp}, when $(A,C)=(q,1)$, the corresponding KdV tau-function gives the partition function of stable quasi-map 
invariants~\cite{CK} of 
a point.\footnote{We are grateful to Bumsig Kim for 
sharing with us his knowledge about this interesting point.}
\end{remark}

\begin{remark}
For the second case of~\eqref{ubisp}, when $(A,B,C)$ is take as $(1,0,\tfrac18)$, the above theorem recovers the formula~\eqref{RTheta}. 
\end{remark}

The following lemma due to Norbury has a similar flavour with Theorem~\ref{bispecthm}.
\begin{lemma} [\cite{N}]
A collection of cohomology classes $\bigl\{\Theta_{g,n}\in H^*\bigl(\overline{\mathcal{M}}_{g,n}\bigr)\bigr\}_{2g-2+n>0}$ satisfying 
the properties~i) and~ii) (see equations~\eqref{np1} and~\eqref{np2}) must satisfy 
$$
\mbox{either} \quad {\rm deg} \, \Theta_{g,n} \,\equiv\,0 \,, \qquad \mbox{or} \quad {\rm deg} \, \Theta_{g,n} \,\equiv\, 2g-2+n\,.
$$
\end{lemma}
\pf
Denote $d(g,n)={\rm deg} \, \Theta_{g,n}$. Using the first equation of \eqref{np2} we obtain 
$$
d(g,n) \= d(g-1,n+2) \,
$$
which implies the existence of a function $H$ of one variable such that $d(g,n)=H(2g-2+n)$. Now using the second equation of \eqref{np2} we find
$$H(x+y)\=H(x)\+H(y)\,.$$
Therefore there exists an integer constant $Q$ such that $d(g,n)=(2g-2+n) \, Q$. Noting that 
$$0\leq d(g,n)\leq 3g-3+n\,,$$ we have either $Q=0$ or $Q=1$. 
\epf

\medskip
\medskip
\medskip
\medskip

\noindent Boris Dubrovin

\noindent SISSA, via Bonomea 265, Trieste 34136, Italy


\medskip
\medskip

\noindent Di Yang

\noindent School of Mathematical Sciences, University of Science and Technology of China,

\noindent Jinzhai Road 96, Hefei 230026, P.R. China 

\noindent diyang@ustc.edu.cn

\medskip
\medskip

\noindent Don Zagier

\noindent Max-Planck-Institut f\"ur Mathematik, Vivatsgasse 7, Bonn 53111, Germany, \\
and 
International Centre for Theoretical Physics, Strada Costiera 11, Trieste 34014, Italy

\noindent dbz@mpim-bonn.mpg.de

\end{document}